%% file: private_med_dc.tex
\documentclass[11pt]{article}
\usepackage[a4paper, margin=1in]{geometry}

\usepackage[caption=false]{subfig}
\usepackage[usenames]{color}
\usepackage[figuresright]{rotating}
\usepackage{boxedminipage}
\usepackage[colorlinks,citecolor=blue,urlcolor=blue]{hyperref}
\usepackage{epstopdf}
\usepackage{natbib}
\usepackage{multirow}
\usepackage{enumerate}
\usepackage{enumitem}
\usepackage{verbatim}
\usepackage{bbm}
\usepackage{xcolor}
\usepackage{bm}
\usepackage{mathrsfs,color,dsfont}
\usepackage{algorithm,setspace}
\usepackage{algorithmic}
\usepackage{comment}
\usepackage{soul,xcolor}

\usepackage{tikz}
\usetikzlibrary{arrows}
\usetikzlibrary{decorations.pathreplacing}

\usepackage{amsthm,amsmath,amssymb,amsfonts, amsbsy, mathtools, epsfig}

\newcommand{\blind}{1}

\renewcommand{\baselinestretch}{1.31}
\newcommand{\linsp}{\renewcommand{\baselinestretch}{1.31}}
\newcommand{\linsps}{\renewcommand{\baselinestretch}{1.3}}

\input{notations}

\newcommand{\laplace}{ \text{\normalfont{Lap}} }
\newcommand{\majvote}{\text{\normalfont{MajVote}} }
\newcommand{\dpvote}{\text{\normalfont{DPVote}} }
\newcommand{\peel}{\text{\normalfont{Peeling}} }

\newcommand{\vote}{\text{\normalfont{Vote} }}

\newcommand{\supp}{\mathrm{supp}}	

\definecolor{DSgray}{cmyk}{0,1,0,0}

\definecolor{scolor}{cmyk}{0.5,2,0,0}

\newtheorem{lemma}{Lemma}
\newtheorem{theorem}{Theorem}
\newtheorem{definition}{Definition}

\newtheorem{proposition}{Proposition}
\newtheorem{remark}{Remark}

\begin{document}

\setstcolor{red}


\linsps

\if1\blind
{
	\title{\bf Majority Vote for Distributed Differentially Private Sign Selection}
		\author{Weidong Liu\thanks{
			School of Mathematical Sciences and MoE Key Lab of Artificial Intelligence, Shanghai Jiao Tong University, Shanghai 200240, China.}\,,
		Jiyuan Tu\thanks{
			School of Mathematics, Shanghai University of Finance and Economics, Shanghai 200433, China.
		}\,,
		Xiaojun Mao\thanks{
			School of Mathematical Sciences and Ministry of Education Key Laboratory of Scientific and Engineering Computing, Shanghai Jiao Tong University, Shanghai 200240, China.}\,,
		and Xi Chen
		\thanks{Stern School of Business, New York University, NY 10012, USA.
}
\thanks{Xiaojun Mao and Xi Chen are the co-corresponding authors.}
}

\date{}
	\maketitle
} \fi

\if0\blind
{
	\bigskip
	\bigskip
	\bigskip
	\begin{center}
		{\LARGE\bf Majority Vote for Distributed Differentially Private  Sign Selection}
	\end{center}
	\medskip
} \fi


\bigskip
\begin{abstract}
	Privacy-preserving data analysis has become more prevalent in recent years. In this study, we propose a distributed group differentially private Majority Vote mechanism, for the sign selection problem in a distributed setup. To achieve this, we apply the iterative peeling to the stability function and use the exponential mechanism to recover the signs. For enhanced applicability, we study the private sign selection for mean estimation and linear regression problems, in distributed systems. Our method recovers the support and signs with the optimal signal-to-noise ratio as in the non-private scenario, which is better than contemporary works of private variable selections. Moreover, the sign selection consistency is justified by theoretical guarantees. Simulation studies are conducted to demonstrate the effectiveness of the proposed method.
\end{abstract}


\linsp


\section{Introduction}

Recently, large amounts of sensitive data have been collected in a distributed manner. While one wants to extract more accurate statistical information from distributed datasets, critical awareness is required toward the probability of possible sensitive personal information leakages, during the learning process. This calls for the study of distributed learning under privacy constraints \citep{pathak_rane_bhiksha.2010, hamm_cao_belkin.2016, jayaraman_wang_etal.2018}.

In this study, we consider the following distributed sign recovery problem. 
Suppose that all $N$ data points $\mbX=\{\vect{X}_i\}_{i=1}^N$ are stored in a distributed system,  that consists of $m$ local machines $\mcH_j$ for $1\leq j\leq m$, and each local machine has $n=N/m$ observations. We focus on estimating
the nonzero positions and signs of a $p$-dimensional population parameter $\vect{\theta}^{*}$, with distributed group differential privacy (DGDP, see Definition \ref{def:dist_dp}). Let $\sgn(\vect{\theta}^*)=(\sgn(\theta^*_1),\dots,\sgn(\theta^*_p))$, where each of its elements takes a value in $\{1,-1,0\}$.  In non-distributed and non-privacy settings,  estimating $\sgn(\vect{\theta}^*)$, which is usually referred to as a sign/variable selection problem, is an important topic in high-dimensional statistics. Many sign/variable selection problems, including estimating the location of the nonzero elements in mean vectors or covariance matrices, support recovery of regression coefficients, and structural estimation of graphical models, have been well studied in the statistical community (see, e.g., \cite{butucea_etal.2018aos, butucea_etal.2020arXiv, tibshirani1996regression, miller.2002, meinshausen2006,karoui.2008aos, bickel_levina.2008, cai_liu.2011, cai_liu_luo.2011}). The sign selection problem also finds applications in many fields, including anomaly detection, medical imaging, and genomics \citep{fan_li.2001jasa, cai_liu_xia.2014jrssb}. There also exist several well-developed statistical tools such as the thresholding technique, $\ell_1$-regularization, and multiple testing to tackle the sign/variable selection problems. However, these problems become more challenging under differential privacy, without their in-depth exploration in the existing literature.

First, we define DGDP for distributed learning. In the privacy literature, differential privacy (DP), first proposed by \cite{dwork_etal.2006}, has been the most widely adopted definition of privacy tailored to statistical data analysis. Note that the $(\epsilon,\delta)$-differential privacy ($(\epsilon,\delta)$-DP) in \cite{dwork_etal.2006}, is designed to protect the privacy between adjacent datasets with a Hamming distance of one, indicating two datasets differing only by one sample (see, e.g., \cite{lei.2011}). In distributed learning, a local machine needs to protect the privacy of all its data points. It is natural to extend the standard DP to distributed group differential privacy as follows.
\begin{definition}  \label{def:dist_dp}
    (Distributed Group Differential Privacy) In a distributed system, assume that the entire dataset is stored in machines $\mcH_1,\dots,\mcH_m$. A randomized algorithm $\mcA:\mcX^N\rightarrow\Theta$ is $(\epsilon,\delta)$-distributed group differentially private ($(\epsilon,\delta)$-DGDP) if for any pair of datasets $\mbX\in\mcX^N$ and $\mbX'\in\mcX^N$, whose elements are the same except for one local machine, the following holds
	\begin{equation*}	
		\mbP\big\{\mcA(\mbX)\in U\big\}\leq e^{\epsilon}\cdot\mbP\big\{\mcA(\mbX')\in U\big\}+\delta,
	\end{equation*}
	for every measurable subset $U\subseteq\Theta$.
\end{definition}

Note that, when the sample size of each local machine is $n=1$, the DGDP reduces to the classical $(\epsilon,\delta)$-DP. Although there is in-depth literature on estimating the parameter $\vect{\theta}^{*}$ under the DP constraint, an efficient procedural design for DGDP sign selection in distributed systems is highly nontrivial. However, there are three major challenges in this regard. First, naively adding noise (a common DP technique) to the output of a non-private sign/variable selection method, makes it difficult to maintain the desired sign selection consistency. Second, most of the existing methods for parameter estimation with DP, are usually inoperative for sign selection. The essential difference between parameter estimation and sign selection is well known. The latter is closely related to multiple testing and top-$k$ selection. To the best of our knowledge, there is insufficient literature on DGDP multiple testing in distributed systems. On the other hand, the top-$k$ selection aims to select the largest $k$ elements among the $p$ given values, under the DP constraint (see, e.g., \cite{bafna_ullman.2017colt, steinke_ullman.2017focs, dwork_su_li.2018differentially, qiao_su_li.2021icml}). However, top-$k$ selection does not imply sign consistency, since obtaining the prior information about the number of non-zero elements $s$ is impractical, and it is impossible to set $k=s$.
Third, sign selection is relatively easy with a large signal-to-noise ratio (SNR) of the nonzero signals. However, when the SNR is close to the minimax lower bound (typically of the order $O(\sqrt{\log p/N})$), it complicates the theoretical analysis for sign selection consistency. As far as we know, for many important statistical problems including support recovery of linear regression under DP, no method attains the minimax lower bound even in a non-distributed setting.

In this paper, we address the above challenges by developing a general Majority Vote mechanism for sign selection problems,
which is particularly adaptive to the distributed system. Assume that $m$ parties vote to make a decision. Each party can vote for positive (1), negative (-1), and null (0). The positive (negative) decision will be made only when there are more than half of the votes for it; otherwise, it returns null. We leverage this Majority Vote mechanism to sign recovery problems in a distributed setup. More specifically, let each local machine obtain an initial sign vector based on its local samples and transmits the initial sign vector to a trustworthy server. By carefully constructing the stability function using the Majority Vote mechanism and combining it with the peeling technique and exponential mechanism in \cite{mcsherry_talwar.2007}, we develop a Majority Vote mechanism that protects distributed group differential privacy (DGDP). The server then generates the result based on this DGDP algorithm. Our method has the following advantages. First, every worker machine can conveniently apply classical techniques such as thresholding and/or Lasso to obtain an initial estimate, and every component only takes value in $\{1,-1,0\}$. Since every worker machine only needs to transmit the sign vector, the privacy of data on each machine is protected to a certain degree. Second, the proposed method only incurs very low communication costs as it only transmits the sign vector. More contributions to this work are summarized below.
\begin{itemize}
    \item To the best of our knowledge, we proposed the first distributed group differential privacy-aware sign selection algorithm in a distributed system. Moreover,  the proposed method guarantees sign selection consistency.
    
    \item For both the sparse mean estimation problem and the sparse regression problem, our proposed method allows the minimal signal to have the order $O(\sqrt{\log p/N})$, which meets the optimal ``beta-min" condition \citep{wainwright2009sharp} in the non-private case. In addition, the DP literature‘s often assumed the boundedness condition, is not required by our theory for data samples.

    \item Our Majority Vote mechanism and DGDP guarantee are substantially general, and do not assume any special underlying statistical model. Thus, in addition to the mean model and linear regression, they can be further applied to many other sign selection problems, such as sparse Gaussian graphical model estimation.
\end{itemize}


\subsection{Related Works}

Variable/sign selection and sparse parameter estimation have been fundamental problems in statistics. In recent years, significant efforts have been made to extend sparse learning/variable selection methods to a distributed setup. Several distributed algorithms, such as the distributed version of the proximal gradient method \citep{shi_ling_etal.2015tsp}, alternating direction method of multipliers \citep{boyd_etal.2011admm}, approximate Newton method \citep{wang_kolar_etal.2017}, and the Frank-Wolfe Algorithm \citep{bellet_liang_etal.2015} for sparse learning, have been proposed one after another. Moreover, in \cite{zhu_li_wang.2019}, the authors introduced the distributed Bayesian information criterion for variable selection. Furthermore, in \cite{battey2018distributed} and \cite{lee_liu_etal.2017}, the authors proposed aggregating local debiased Lasso estimators in one shot to achieve near-optimal statistical accuracy. On the contrary, our method directly aggregates local Lasso estimators (instead of debiased Lasso estimators) and achieves the near-optimal minimal signal strength condition, which is a novel contribution to the literature, despite privacy concerns.

In the past decade, differential privacy has gained significant attention in computer science and statistics. Many algorithms have been developed to deal with various problems, such as empirical risk minimization \citep{bassily_smith_thakurta.2014}, principal component analysis \citep{dwork_talwar_etal.2014stoc}, demand learning \citep{chen_simchi_wang.2022, chen_miao_wang.2021arXiv} and deep learning \citep{bu2020deep}. In addition to the classical differential privacy framework, several proposed privacy notions help adapt to different situations. For example, group DP \citep{dwork_aaron.2014}, local DP \citep{wasserman_zhou.2010, duchi_jordan_wainwright.2018}, Gaussian DP \citep{dong_roth_su.2019}, concentrated DP \citep{bun_steinke.2016}, R\'enyi DP \citep{ilya.2017}, and on-average KL-privacy \citep{wang_lei_fienberg.2016}. This study proposes the concept of distributed group DP, which is designed for distributed setup. The group structure in the distributed system is predetermined. This differs from the group DP in \cite{dwork_aaron.2014}, where the DP is relative to arbitrary subgroups of a given size. 

There are many works that study differentially private sparse regression problems, such as \cite{jain_thakurta.2014icml, talwar_etal.2015nips, kasiviswanathan_jin.2016icml, cai_wang_zhang.2019, wang_xu.2019icml}. However, most of these studies considered minimizing the loss function or proving $\ell_2$-consistency. It is not direct to obtain the model selection results from these works. To the best of our knowledge, only three works consider DP variable selection for linear regression in a non-distributed setting: \cite{kifer_smith_thakurta.2012colt}, \cite{thakurta_smith.2013colt}, and \cite{lei_charest_etal.2018jrssa}. We perform a more detailed comparison between these studies and ours. In \cite{kifer_smith_thakurta.2012colt}, two algorithms were proposed. One is based on an exponential mechanism \citep{mcsherry_talwar.2007}. Their algorithm requires solving many optimizations constrained in any $s$-sparse subspace. As pointed out in  \cite{kifer_smith_thakurta.2012colt}, this algorithm is computationally inefficient. The other used a resample-and-aggregate method \citep{kobbi_sofya_adam.2007, adam.2011}, which directly counts the number of nonzero elements that are initially estimated from $m$ block of subsamples and applies $\peel$ based on this counting. Their algorithm ignores the sign information from the initial estimators, which is useful for identifying the zero positions. Moreover, their method does not result in sign selection consistency and requires that the nonzero elements have magnitudes larger than $O(\sqrt{s\log p}/N^{1/4})$, which is not minimax optimal. In \cite{thakurta_smith.2013colt}, the authors defined two concepts of stability and proposed a PTR-based mechanism for variable selection. This method has a nontrivial probability of outputting the null (no result), which is undesirable in practice. In theory, it requires the boundness of the true parameter $\vect{\theta}^*$ and the covariates $\vect{X}$, which is not needed in our method. Moreover, they require the signals to be larger than $O(\sqrt{s\log p/N})$, which is not minimax optimal.
Furthermore, their algorithm was designed for differential privacy rather than DGDP. In \cite{lei_charest_etal.2018jrssa}, the authors proposed the use of the Akaike or Bayesian information criterion, in conjunction with the exponential mechanism, to choose the proper model. However, their method traverses all possible models and causes a heavy computational burden.

Variable selection is also related to multiple testing, which can be used to find significant signals while controlling the false discovery rate (FDR). \cite{dwork_su_li.2018differentially} proposed a differentially private procedure to control the FDR in multiple hypothesis testing. However, it cannot be applied directly to sign selection in a distributed group differential privacy setting. 

Our method also shares a similar spirit with the subsample-and-aggregate (SA) methods in the DP literature \citep{kobbi_sofya_adam.2007, adam.2011, kifer_smith_thakurta.2012colt}. Like SA methods, we compute noiseless results for each block and then aggregate them in a DP manner. However, our method differs from existing SA methods in several aspects. In our model, the data are already distributed, and our method ensures the privacy of each entire block, as opposed to existing SA methods that only protect the privacy of each element;  SA was originally proposed in \cite{kobbi_sofya_adam.2007} to compute the smooth sensitivity of statistics, which has a different motivation from our method; In \cite{adam.2011}, SA was advocated to achieve the same efficiency as non-private statistics. However, their method is limited to low-dimensional statistics and does not extend well to high-dimensional problems like model selection. In \cite{kifer_smith_thakurta.2012colt}, SA was employed for model selection with a voting mechanism as the aggregator. In contrast, our method uses a different voting rule, and our theoretical results are much stronger than theirs. For a comprehensive comparison of theoretical results, please refer to the preceding two paragraphs.


\subsection{Paper Organization and Notations}
The remainder of this paper is organized as follows. In Section \ref{sec:ptr_recov}, we introduce a general Majority Vote procedure and develop the DGDP algorithm. In Section \ref{sec:mean_recov}, we apply this to the private sign recovery problem for sparse mean estimation, in a distributed system. The theoretical results on sign selection consistency are provided.  In Section \ref{sec:lasso_recov}, we study the sign recovery problem for sparse linear regression and simultaneously state the sign selection consistency results. The outcomes of simulation studies are presented in Section \ref{sec:sim}, to demonstrate the proposed method's effectiveness. Finally, concluding remarks are provided in Section \ref{sec:conclude}. All proofs of the theory are relegated to the Appendix.

For every vector $\vect{v}=(v_1,...,v_p)^T$, define $|\vect{v}|_2=\sqrt{\sum_{l=1}^pv_l^2}$, $|\vect{v}|_1=\sum_{l=1}^p|v_l|$, and $|\vect{v}|_{\infty}=\max_{1\leq l\leq p}|v_l|$. Moreover, we use $\supp(\vect{v})=\{l\mid v_l\neq 0, 1\leq l\leq p\}$ to denote the support of the vector $\vect{v}$, and define $\vect{v}_{-l}=(v_1,\dots,v_{l-1},v_{l+1},\dots,v_p)^{\tp}$. For each matrix $\vect{A}\in\mbR^{p_1\times p_2}$, $\Norm{\vect{A}}=\sup_{|\vect{v}|_2=1}|\vect{A}\vect{v}|_2$, $\Norm{\vect{A}}_{\infty}=\max_{1\leq l_1\leq p_1,1\leq l_2\leq p_2}|A_{l_1,l_2}|$ and $\Norm{\vect{A}}_{L_{\infty}}=\sup_{|\vect{v}|_{\infty}=1}|\vect{A}\vect{v}|_{\infty}$. Furthermore, let $\Lambda_{\max}(\vect{A})$ and $\Lambda_{\min}(\vect{A})$ denote the largest and smallest eigenvalues of $\vect{A}$, respectively. We use $\mbI(\cdot)$ to denote the indicator function, and $\sgn(\cdot)$ to denote the sign function. For two sequences $\{a_n\},\{b_n\}$, we say $a_n\asymp b_n$ if $a_n=O(b_n)$ and $a_n = \Theta(b_n)$ hold simultaneously. For simplicity, we let $\mbS^{p-1}$ and $\mbB^{p}$ denote the unit sphere and the unit ball in $\mathbb{R}^p$ centered on $\vect{0}$. For a sequence of vectors $\{\vect{v}_i\}_{i=1}^n\subseteq\mbR^p$, we define $\med(\cdot)$ as the coordinate-wise median. Finally, the generic constants were assumed to be independent of $m,n,$ and $p$. We define $[p]$ to be the set $\{1,\dots,p\}$.  


\section{A General  Majority Vote Procedure for Privacy Preserving Sign Recovery} \label{sec:ptr_recov}

In this section, we introduce the Majority Vote procedure in a distributed system. Suppose that local machines estimate sgn($\vect{\theta}^{*}$) using the local dataset $\{\vect{X}_i,i\in\mcH_{j}\}$ separately and send the results to a trustworthy server. The server then aggregates the initial estimators received and provides a final result to the user. In this process, two pivotal procedures require careful design. First, the aggregation algorithm on the server integrates the sign information from initial estimators as much as possible. Moreover,  a specific noise-addition technique was developed for DGDP. This requires an innovative way to integrate existing DP methods. Second, local machines require a feasible approach to generate good initial estimators. A natural strategy to address this problem is the regularization technique. However, in the distributed setting, we argue in Section \ref{sec:mean_recov} that the classical choice of regularization parameters can not attain the minimax optimal rate at the signal level. Therefore, careful design of the regularization parameter is important.

\subsection{ A Majority Vote Procedure}

First, we formally introduce the Majority Vote procedure. Let $\vect{Q}_j^c$ be the sign vector obtained by the local machine $\mcH_j$ based on $\{\vect{X}_i,i\in\mcH_{j}\}$ (where $1\leq j\leq m$). The server then has a $p\times m$ sign matrix $\mbQ=(Q_{l,j}) = (\vect{Q}^c_1,\dots,\vect{Q}^c_m)$ where each $Q_{l,j}$ takes a value of $\{1,-1,0\}$, representing positive, negative and null, respectively. Here $\vect{Q}^c_j=(Q_{1,j}, \dots,Q_{p,j})^{\tp}$ denotes each column vector received from the $j$-th local machine (where $1\leq j\leq m$). For each coordinate $1\leq l\leq p$, we consider the Majority Vote with the row vector $\vect{Q}_l^{r}=(Q_{l,1},\dots,Q_{l,m})$. The mechanism is as follows. When more than half of the elements in $\vect{Q}_l^r$ have values of $1$ (or $-1$), the output is this value. Otherwise, it returns null ($0$). Formally, for each row $\vect{Q}^r_l$, we compute the number of positives, negatives, and nulls as follows:
\begin{equation}	\label{eq:pnz_number}
	N^+_l=\sum_{j=1}^m\mbI(Q_{l,j}=1),\quad N^-_l=\sum_{j=1}^m\mbI(Q_{l,j}=-1),\quad N^0_l=\sum_{j=1}^m\mbI(Q_{l,j}=0).
\end{equation}
Then $N_l^++N_l^-+N_l^0=m$ always holds, and the Majority Vote of $\vect{Q}^r_l$ can be equivalently computed as
\begin{equation*}
	\majvote\{\vect{Q}^r_l\}=
	\begin{cases}
		1\quad&\text{ if }N_l^+\geq N_l^0+N_l^-+1,\\
		-1\quad&\text{ if }N_l^-\geq N_l^0+N_l^++1,\\
		0\quad&\text{ otherwise}.
	\end{cases}
\end{equation*}
Let $\bar{\vect{Q}}$ be the  resulting vector of $\mbQ$ by the Majority Vote, that is,
\begin{eqnarray}\label{eq:tt}
    \bar{\vect{Q}} = (\majvote\{\vect{Q}^r_l\}, 1\leq l\leq p)^{\rm T}.
\end{eqnarray}
The merit of the Majority Vote will be illustrated in Section 3.1.  For the mean model and linear regression model, we show that
the vector $\bar{\vect{Q}}$ has sign selection consistency under some weak conditions (see Propositions \ref{thm:priv_mean_supp} and \ref{thm:priv_reg_supp} in the Appendix).
Although $\bar{\vect{Q}}$ has already protected the privacy of the local dataset in some sense (note that local machines only send $\{1,-1,0\}$ to the server), it does not satisfy DGDP in Definition \ref{def:dist_dp}. To solve this problem,   we integrate the Majority Vote with the Peeling algorithm and exponential mechanism and develop a distributed group differentially private algorithm. Furthermore, we prove that the new algorithm can still exhibit sign selection consistency. 

Before introducing our DGDP algorithm, we first review some classical definitions of the sensitivity function, exponential mechanism, and composition theorem for differential privacy.
For simplicity, we define $\bar{S}$ as the index set of nonzero elements in $\bar{\vect{Q}}$ and let $\bar{s}=|\bar{S}|$.


\subsection{Privacy Preliminaries}  \label{sec:priv_pre}

Popularly, a key step to developing a differentially private algorithm is to impose additional randomness on the selection procedure. The scale of such randomness is typically determined by the global sensitivity of the algorithm.

\begin{definition}[Global Sensitivity]
	Given an algorithm $\mcA:\mcX^N\rightarrow\Theta$, the global sensitivity of $\mcA$ is defined as
	\begin{equation*}
		\mathrm{GS}_{\mcA}=\sup\big\{|\mcA(\mbX)-\mcA(\mbX')|:\mbX,\mbX'\in\mcX^N, H(\mbX,\mbX')=1\big\},
	\end{equation*}
	where $H(\cdot,\cdot)$ denotes the Hamming distance.
\end{definition}

For the global sensitivity for DGDP in the distributed setting, we let $H(\mbX,\mbX')=1$ denote that all elements are the same except for one local machine.    
Global sensitivity measures the magnitude of the change in the output of $\mcA$, resulting from replacing one element (or one group of elements) in the input dataset. Intuitively, when the noise is scaled proportionally to the global sensitivity, the output of the differentially private version of $\mcA$ is relatively stable, regardless of the presence or absence of any individual data in the dataset. 
Indeed, we have the following result.
\begin{lemma}[The Laplace mechanism, Theorem 3.6 of \cite{dwork_aaron.2014}]
    For any algorithm $\mcA$ satisfying $\mathrm{GS}_{\mcA}<\infty$, $\mcA_1=\mcA + g$, where $g$ is sampled from $\laplace(\mathrm{GS}_{\mcA}/\epsilon)$, achieves $(\epsilon,0)$-differential privacy.   
\end{lemma}

Here, $\laplace(b)$ denotes the Laplace distribution with the density function $\frac{1}{2b}\exp(-|x|/b)$ for $x\in (-\infty,\infty)$, and the scale parameter $b>0$.
Our algorithm acts on every coordinate as the composition of multiple actions, aiding us to introduce two composition theorems for the convenient construction of our algorithm.

\begin{lemma}[Composition theorem, Theorem B.1 of \cite{dwork_aaron.2014}]	\label{lem:two_comp}
	Let $\mcA_1:\mcX^n\rightarrow\Theta_1$ be an $(\epsilon_1,\delta_1)$-differentially private mechanism, and $\mcA_2:(\mcX^n,\Theta_1)\rightarrow\Theta_2$ be an $(\epsilon_2,\delta_2)$-differentially private mechanism for every fixed element $\theta_1\in\Theta_1$. Then, the composition of the mappings, $(\mcA_1,\mcA_2):\mcX^n\rightarrow(\Theta_1,\Theta_2)$ by mapping $\mbX\in\mcX^n$ to $(\mcA_1(\mbX),\mcA_2(\mbX,\mcA_1(\mbX)))$, is an $(\epsilon_1+\epsilon_2,\delta_1+\delta_2)$-differentially private mechanism.
\end{lemma}

\begin{lemma}[Advanced composition theorem, Corollary 3.21 in \cite{dwork_aaron.2014}]	\label{thm:adv_comp}
	For all $\epsilon\in(0,1),\delta,\delta'\in[0,1]$, the $k$-fold adaptive composition of the class of $(\epsilon,\delta)$-differentially private mechanisms preserves $(2\sqrt{2k\log(1/\delta')}\epsilon, k\delta+\delta')$-differential privacy.
\end{lemma}

We refer to Section 3.5.2 of \cite{dwork_aaron.2014} for a more detailed introduction to the definition of differential privacy, under the $k$-fold adaptive composition. Invariably, the two aforementioned composition theorems still hold for the DGDP.

Note that our goal is to recover the sign vector, whereas the direct addition of Laplace noise to $\bar{\vect{Q}}$ can completely destroy its value.
 It is well known that when the algorithm takes values with some special structure $\Theta$ (e.g. $\Theta=\{1,-1,0\}$ in sign recovery), we can apply the exponential mechanism invented by \cite{mcsherry_talwar.2007}. 
 The exponential mechanism maintains the same structure as $\Theta$ for the final output. In more detail,
 this mechanism assigns each pair of values in $\Theta$ and the dataset $\mbX$, with a utility function $f:\mcX^n\times\Theta\rightarrow\mbR$, and generates a noisy output according to the global sensitivity of the utility function. In particular, we define 
\begin{eqnarray}\label{ut}
	\Delta f = \max_{\theta\in\Theta}\max_{H(\mbX,\mbX')=1}|f(\mbX,\theta)-f(\mbX',\theta)|.
\end{eqnarray}
The exponential mechanism outputs a result $\theta\in\Theta$ with probability proportional to $\exp(\epsilon f(\mbX,\theta)/(2\Delta f))$, that is,
\begin{eqnarray*}
\mbP\Big{(}\text{OUTPUT}=\theta\Big{)}=c\exp(\epsilon f(\mbX,\theta)/(2\Delta f)),
\end{eqnarray*}
where $c$ is the scale constant of the probability distribution.
We refer the reader to Section 3.4 in \cite{dwork_aaron.2014} for more details about the exponential mechanism. The following lemma provides the privacy guarantee for the exponential mechanism.

\begin{lemma}(Theorem 3.10 in \cite{dwork_aaron.2014})	\label{lem:exp_mech}
	The exponential mechanism preserves $(\epsilon,0)$-differential privacy.
\end{lemma}

The exponential mechanism is also $(\epsilon,0)$-DGDP when we modify $H(\mbX,\mbX')=1$ in (\ref{ut}) to denote that all elements are the same, except for one local machine. In fact, we can view the samples in one local machine as one higher-dimensional vector. Thus, all the conclusions on DP hold naturally for DGDP.
The above lemma provides us with a method to modify the Majority Vote to a differentially private counterpart.
 We must leverage the knowledge of the sparseness of the true parameter. If we ignore the sparseness and apply the DP algorithm to every coordinate, by Lemma \ref{thm:adv_comp}, we know that each coordinate should preserve $(\tilde{\epsilon},\tilde{\delta})$ privacy, where $\tilde{\epsilon}\asymp \frac{\epsilon}{\sqrt{p\log(1/\delta)}}$. The noise level imposed on each coordinate should then be scaled as $\frac{\sqrt{p\log(1/\delta)}}{\epsilon}$, which may significantly deteriorate the performance when $p$ is large. Therefore, it is important to use sparseness to reduce the level of noise and improve performance. To this end, we use the Peeling algorithm, which is typically used in the problem of differentially private top-k selection and was also used in \cite{dwork_aaron.2014, dwork_su_li.2018differentially, cai_wang_zhang.2019}. Using the Peeling algorithm, we can reduce the dimension and select a subset of $\{1,2,...,p\}$ that asymptotically covers the true support. The challenge is to
construct an efficient utility function $f$ that is required in both the Peeling algorithm and exponential mechanism. That is, we must construct a measure of the closeness of the dataset $\mbQ$ to the discrete set $\{-1,0,1\}$. In the next section, we solve these problems by introducing the concept of a stability function.


\subsection{Private Sign Selection According to Majority Vote}	\label{sec:medsign}	

Recall the definitions of $\{N^+_l,N^+_l,N^0_l\}$ given in \eqref{eq:pnz_number}.
To solve the questions addressed at the end of the previous section, for each row $\vect{Q}^r_l$ where $1\leq l\leq p$, we define the stability level of dataset $\vect{Q}^r_l$ with respect to the $\majvote(\cdot)$ function as the minimal number of elements in $\vect{Q}^r_l$ which need to be flipped to change the value of $\majvote\{\vect{Q}^r_l\}$. This can also be explicitly computed as:
\begin{equation}	\label{eq:stab_fun}
	f^S(\vect{Q}^r_l)=
	\begin{cases}
	N^+_l- N^0_l-N^-_l\quad&\text{if }\majvote\{\vect{Q}^r_l\}=1,\\
	N^-_l- N^0_l-N^+_l\quad&\text{if }\majvote\{\vect{Q}^r_l\}=-1,\\
	-\min\big\{N^+_l+ N^0_l-N^-_l, N^-_l+ N^0_l-N^+_l\big\} \quad&\text{if }\majvote\{\vect{Q}^r_l\}=0.
	\end{cases}
\end{equation} 
Note that we assign the stability $f^S(\vect{Q}^r_l)$ with a minus sign for $\majvote\{\vect{Q}^r_l\}=0$ because we only want to select the nonzero elements.

\begin{algorithm}[!t]
	\caption{{\small Peeling ($\peel(\mbQ,\tilde{s},\epsilon,\delta)$)}}
	\label{alg:suppselect}
	\hspace*{\algorithmicindent} \hspace{-0.5cm}   {\textbf{Input:} The set of signs $\mbQ=\{\vect{Q}^c_1,...,\vect{Q}^c_m\}$; the number of selections $\tilde{s}$; privacy level $(\epsilon,\delta)$; initial $\tilde{S}=\emptyset$.} 	
	\begin{algorithmic}[1]
		\FOR{$1\leq l\leq p$}
		\STATE Compute the stability level $f^S(\vect{Q}_l^r)$ based on \eqref{eq:stab_fun}.
		\ENDFOR
		\FOR{$1\leq t\leq \tilde{s}$}
		\STATE Generate $g_{t,1},\dots,g_{t,p}\sim\laplace(4\sqrt{2\tilde{s}\log(1/\delta)}/\epsilon)$;
		\STATE Add $l^*=\argmax{l\in[p]\backslash\tilde{S}}f^S(\vect{Q}_l^r)+g_{t,l}$ to $\tilde{S}$.
		\ENDFOR
	\end{algorithmic}
	 \textbf{Output:} Return the sets $\tilde{S}$.
\end{algorithm}
To select a subset that covers the support of $\bar{\vect{Q}}$ in a private way, we leverage the $\peel$ algorithm with the stability function $f^S(\vect{Q}_l^r)$.  This procedure is presented in Algorithm \ref{alg:suppselect}. Note that in this algorithm we repeatedly select the maximal element in a private manner. Recently, \cite{qiao_su_li.2021icml} proposed a one-shot $\mathrm{Top}$-$\tilde{s}$ selection algorithm that reduces computational burden. However, it uses a larger scale of noise, and theoretically only accepts small $\epsilon$ and $\delta$ (namely, $\epsilon\leq 0.2$ and $\delta\leq 0.05$). Thus, it is not adopted in our study.

The Peeling algorithm finds a subset $\tilde{S}$ that asymptotically covers the support of $\bar{\vect{Q}}$. Notably, this does not generate the signs for each coordinate.  Therefore, after generating the set $\tilde{S}$ for each coordinate $l$ (where $1\leq l\leq p$), we apply the exponential mechanism by defining the following utility function:
\begin{equation}    \label{eq:exp_util}
    \begin{cases}
        f_l(\vect{Q}_l^r,1)=N_l^+-N_l^0-N_l^-,\\
        f_l(\vect{Q}_l^r,-1)=N_l^--N_l^0-N_l^+,\\
        f_l(\vect{Q}_l^r,0)=\min\big\{N^+_l+ N^0_l-N^-_l, N^-_l+ N^0_l-N^+_l\big\}.
    \end{cases}
\end{equation}
Using this utility function, we can easily verify that the global sensitivity in (\ref{ut}) is $\Delta f=2$. That is, changing the entire dataset in one local machine can at most change the value of $f$ by $\pm 2$. This results in the realization of our private Majority Vote procedure as presented in Algorithm \ref{alg:medsign}.

\begin{algorithm}[!t]
	\caption{{\small Differentially Private Majority Vote for Sign Recovery ($\dpvote(\mbQ,\tilde{s},\epsilon,\delta)$)}}
	\label{alg:medsign}
	\hspace*{\algorithmicindent} \hspace{-0.5cm}   {\textbf{Input:} The set of signs $\mbQ=\{\vect{Q}^c_1,...,\vect{Q}^c_m\}$; the number of selections $\tilde{s}$; privacy level $(\epsilon,\delta)$.} 	
	\begin{algorithmic}[1]
		\STATE Apply $\peel(\mbQ,\tilde{s},\epsilon/2,\delta/2)$ to select the index set $\tilde{S}=\{l_1,l_2,\dots,l_{\tilde{s}}\}$.
		\FOR{$l\in\tilde{S}$}
		\STATE Compute the quantities $f_l(\vect{Q}_l^r,1), f_l(\vect{Q}_l^r,-1), f_l(\vect{Q}_l^r,0)$ based on \eqref{eq:exp_util}.
		\STATE Generate the random sign $\hat{Q}_l$ according to the following distribution 
		\begin{equation*}
			\begin{cases}
			\mbP(\hat{Q}_l=1)=\frac{P_l^+}{P_l^+ + P_l^- + P_l^0},\\
			\mbP(\hat{Q}_l=0)=\frac{P_l^0}{P_l^+ + P_l^- + P_l^0},\\
			\mbP(\hat{Q}_l=-1)=\frac{P_l^-}{P_l^+ + P_l^- + P_l^0},
		\end{cases}
		\quad\text{ where }
		\begin{cases}
			&P_l^+ = \exp\{\epsilon' f_l(\vect{Q}_l^r,1)/4\},\\
			&P_l^0 = \exp\{\epsilon' f_l(\vect{Q}_l^r,0)/4\},\\
			&P_l^- = \exp\{\epsilon' f_l(\vect{Q}_l^r,-1)/4\}.
		\end{cases}
		\end{equation*}
		and $\epsilon' = \epsilon/(4\sqrt{2\tilde{s}\log(2/\delta)})$.
		\ENDFOR
	\end{algorithmic}
	 \textbf{Output:} Return the sets $\hat{S} = \{l|l\in\tilde{S},\hat{Q}_l\neq0\}$ and $\hat{Q}=\{\hat{Q}_l|l\in\hat{S}\}$.
\end{algorithm}

The concept of stability in differential privacy is related to the classical ``Propose-Test-Release'' mechanism. The ``Propose-Test-Release'' mechanism, or the PTR mechanism, was first proposed in \cite{dwork_lei.2009}. We also refer the reader to \cite{thakurta_smith.2013colt, dwork_aaron.2014, vadhan.2017} for further exploration of this topic. The PTR mechanism uses stability to measure the distance of the database from the sensitive dataset and decides whether to halt the algorithm. We leverage this idea by using the stability function as the utility function in the $\peel$ algorithm and the exponential mechanism. Unlike the PTR mechanism, our algorithm does not halt and always outputs a non-empty set of signs. 


\subsection{Theories of $\dpvote$ Method}	\label{sec:medsign_theory}

We show that the $\peel$ algorithm and the coordinate-wise exponential mechanism both attain $(\frac{\epsilon}{2},\frac{\delta}{2})$-DGDP. Then, from the composition theorem in Lemma \ref{lem:two_comp}, we know that the entire procedure is $(\epsilon,\delta)$-DGDP. Namely, we have the following theorem:

\begin{theorem}[Differential privacy of $\dpvote$]	\label{thm:medsign_dp}
	The proposed $\dpvote$ mechanism in Algorithm \ref{alg:medsign} is $(\epsilon,\delta)$-distributed group differentially private.
\end{theorem}

Note that the $(\epsilon,\delta)$-DGDP conclusion for our algorithm holds for any sign matrix, $\mbQ$. 
Next, we show that the algorithm outputs the same signs as $\bar{\vect{Q}}$ in (\ref{eq:tt}) asymptotically under weak conditions.
 Let $\bar{S}$ be the support for $\bar{\vect{Q}}$ and $\bar{s}=|\bar{S}|$.

\begin{theorem}[Sign consistency of $\dpvote$]	\label{thm:medsign_cons}
	Given the privacy level $(\epsilon,\delta)$, suppose $\tilde{s}\geq \bar{s}$, and for $1\leq l\leq p$, the stability function satisfies
	\begin{equation}	\label{eq:medsign_cons_term1}
		f_l(\vect{Q}_l^r,\bar{Q}_l) \geq 8(\gamma+1)\sqrt{2\tilde{s}\log(2/\delta)}\log p/\epsilon,
	\end{equation} 
	where $\gamma>2$. Then, the proposed $\dpvote$ algorithm produces $(\hat{S},\hat{\vect{Q}})$ that has the same signs as $\bar{\vect{Q}} $with a probability (with respect to the randomness in the algorithm) of no less than $1-O(p^{-\gamma+2})$. 	
\end{theorem}

In this paper, we focus on the high-dimensional scenario where the dimension $p$ tends to infinity. Therefore, Theorem \ref{thm:medsign_cons} implies consistency under this assumption.
Condition (\ref{eq:medsign_cons_term1})  indicates that, if a sign  (c.f. $+1$) obtains sufficient votes (exceeding the other two  $8(\gamma+1)\sqrt{2\tilde{s}\log(2/\delta)}\log p/\epsilon$ votes),  then the algorithm outputs the true sign of $\bar{\vect{Q}}$ with high probability. The term $8(\gamma+1)\sqrt{2\tilde{s}\log(2/\delta)}\log p/\epsilon$ is determined by the noise level in the Peeling algorithm. Recall that $\bar{\vect{Q}}$ is the estimator (without noise addition) for the sign vector of the population parameter. The proof shows that $\bar{\vect{Q}}$ has sign consistency for the mean vector and linear regression models, which indicates that $\dpvote$ can have sign consistency for these two models. In these settings, we can deduce the corresponding minimal signal strength conditions (see (\ref{eq:mean_mu_min}) and (\ref{eq:regression_theta_min})) from (\ref{eq:medsign_cons_term1}), which are regular in variable selection.

Notably, Theorem \ref{thm:medsign_cons} does not assume any underlying model for the dataset in the local machines. That is, $\dpvote$ works for many statistical models once we can construct statistics in local machines to ensure sign consistency for $\bar{\vect{Q}}$. Therefore, except for the mean model and linear regression problem given in Sections \ref{sec:mean_recov} and \ref{sec:lasso_recov}, $\dpvote$ can also be applied to other variable selection problems such as Gaussian graphical model estimation.

 In the following sections, we apply the proposed $\dpvote$ algorithm to the sign selection for the mean vector and linear regression.


\section{Private Sign Selection of Mean Vector}	\label{sec:mean_recov}

Private mean estimation is a fundamental problem in private statistical analysis and has been intensively studied \citep{dwork_etal.2006, lei.2011, bassily_smith_thakurta.2014, cai_wang_zhang.2019}. The standard approach is to project the data onto a known bounded domain and then apply the noises according to the diameter of the feasible domain and the privacy level. This requires that the input data or true parameter lie in a known bounded domain, which seems unsatisfactory in practice. Furthermore, few studies have examined the sign selection for the mean vector under DGDP. 
As the first application of the $\dpvote$ algorithm, in this section, we investigate the private sign recovery problem for a sparse mean vector in a distributed system. The results of sign selection consistency and privacy guarantee are provided.


\subsection{$\dpvote$ for Mean Estimation}    \label{sec:majvote_mean}

Let $\vect{\theta}^*=(\theta^*_1,\dots,\theta^*_p)^{\tp}$ be the true population parameter of interest. We denote $S$ as the support of $\vect{\theta}^*$ and $s=|S|$. We assume that the vector is sparse in the sense that many entries $\theta^*_l$ are zero. There are $N$ i.i.d. observations $\vect{X}_i$ that satisfy $\mbE[\vect{X}_i]=\vect{\theta}^*$ and are stored in $m$ machines $\mcH_j$ (where $1\leq j\leq m$), each holds $n_j$ (where $1\leq j\leq m$) samples. We assume that $n_j$'s have the same order as $n$, namely $n_1\asymp ...\asymp n_m\asymp n$. We define $\mbX=\{\vect{X}_1,\dots,\vect{X}_N\}$ as the full dataset.  Our task is to identify $\sgn(\vect{\theta}^*)$ under DGDP.

To present our method more clearly, we define the quantization function $\mcQ_{\lambda}(\cdot)$ as follows:
\begin{equation}	\label{eq:quant_def}
	\mcQ_{\lambda}(x)=\sgn\big\{\underbrace{\sgn(x)\cdot(|x|-\lambda)_+}_{\text{shrinkage operator}}\big\}=
	\begin{cases}
		\sgn(x)\quad&\text{ if }|x|>\lambda,\\
		0\quad&\text{ if }|x|\leq\lambda.
	\end{cases}
\end{equation}
Here $\lambda$ is the threshold parameter. When $x$ is a vector, $\mcQ_{\lambda}(x)$ performs the aforementioned operation in a coordinated manner. In particular, when $\lambda=0$, the function $\mcQ_{0}(\cdot)$ is the same as sign function $\sgn(\cdot)$. We then present our method in Algorithm \ref{alg:sign_mean}. The choice of thresholding parameter is discussed after Theorem \ref{prop:dp_mean_sign} in Section \ref{sec:priv_mean_theory}.

\begin{algorithm}[!t]
	\caption{{\small Differentially private Majority Vote ($\dpvote$) for sparse mean.}}
	\label{alg:sign_mean}
	\hspace*{\algorithmicindent} \hspace{-0.6cm} {\textbf{Input:} Dataset $\mbX=\{\vect{X}_1,\dots,\vect{X}_N\}$ stored in $m$ local machines (where $j=1,\dots,m$), the universal thresholding parameter $\lambda_N$, number of selections $\tilde{s}$, privacy level $(\epsilon,\delta)$.} 	
	\begin{algorithmic}[1]
		\FOR{$j=1,\dots,m$}
		\STATE Compute the local sample mean $\bar{\vect{X}}_j=n_j^{-1}\sum_{i\in\mcH_j}\vect{X}_i$ on the $j$-th machine and obtain the sign vector $\vect{Q}_j=\mcQ_{\lambda_N}(\bar{\vect{X}}_j)$. Send $\vect{Q}_{j}$ to the server.
		\ENDFOR
		\STATE Apply $\dpvote$ 
		\begin{equation}	\label{eq:dpmed-mean_sign}
			\hat{\vect{Q}}(\mbX)=\dpvote(\{\vect{Q}_j\}_{1\leq j\leq m},\tilde{s},\epsilon,\delta).
		\end{equation}
	\end{algorithmic}
	\textbf{Output:}  The sign vector $\hat{\vect{Q}}(\mbX)$.
\end{algorithm}

\input{plot/med-dc_intuition.tex}

\noindent{\bf Majority Vote vs direct thresholding.} Before presenting the theoretical results, we first briefly discuss the advantages of the Majority Vote in sign/variable selection. Using the thresholding estimator as the initial statistic,  we let
\begin{equation}	\label{eq:med-mean_sign}
	\bar{\vect{Q}}=\bar{\vect{Q}}(\mbX)=\majvote(\mcQ_{\lambda_N}(\bar{\vect{X}}_j)\mid 1\leq j\leq m),
\end{equation}
where $\bar{\vect{X}}_{j}$'s are the local sample means. To estimate support $S$,
we claim that the Majority Vote procedure is more efficient than direct thresholding on local sample means.
Suppose that the sample $\vect{X}_{i}\sim \vect{\theta}^*+\mcN(0,\vect{I})$. The local sample size is $n$. It is well known that a common threshold level for the local sample mean is $\lambda_{n}=c\sqrt{\log p/n}$ with $c=\sqrt{2}$, such that the local machine can identify all zero positions in $\theta^*$. It is easy to see that any smaller thresholding constant $c<\sqrt{2}$ results in false positives. With the thresholding level $\lambda_{n}$, the minimal signal of nonzero position should be no less than the order $O(\sqrt{\log p/n})$, otherwise, selection consistency becomes impossible. In contrast, the Majority Vote procedure allows for a thresholding level smaller than $\sqrt{2\log p/n}$. In fact, by using
the Majority Vote procedure, we let $\lambda_{N}=c\sqrt{\log p/N}$ for some $c>0$. Recall that $N$ denotes the total number of samples. This thresholding level is much smaller than $\sqrt{2\log p/n}$, and hence, many false nonzero positions are misidentified by local machines. However, due to the symmetry of the normal distribution, the number of false $1$'s and $-1$'s are likely equal, and both are less than half. Thus, after applying the Majority Vote rule, the output is still $0$. Now with a smaller $\lambda_{N}$, we allow a weaker minimal signal strength assumption for nonzero positions. The mechanism of $\majvote$ is illustrated in Figure \ref{fig:med-dc_intuition}.


\subsection{Theory of Mean Vector Sign Selection}	\label{sec:priv_mean_theory}

To discuss the theoretical properties of our method, we introduce the distribution space $\mcP$ for the population $\vect{X}$:
\begin{equation}	\label{eq:sample_dist}
\begin{aligned}
	\mcP(\vect{\theta}^*, C)=\Big\{\mbP\Big|&  \mbE_{\mbP}[\vect{X}]=\vect{\theta}^*, \max_{1\leq l\leq p}\mbE_{\mbP}\big[|X_l-\theta_l^*|^3\big]\leq C \Big\},
\end{aligned}
\end{equation}
where $C>0$ is a constant. This is a relatively weak condition for the distribution of $\vect{X}$. 
 We obtain the following sign consistency result.

\begin{theorem} \label{prop:dp_mean_sign}
    (Sign consistency of $\dpvote$ for the mean vector) Let $\{\vect{X}_1,\dots,\vect{X}_N\}$ be $N$ i.i.d. random vectors sampled from $\mbP(\vect{\theta}^*,C)$ distributed in $m$ machines $\mcH_1,\dots,\mcH_m$, with each machine holds $n_j$ ($1\leq j\leq m$) samples. Here we assume $n_1\asymp ...\asymp n_m\asymp n$. Moreover, assume that there exist sufficiently large constants $C_1,C_2,\gamma_0>0$ such that
    \begin{itemize}
        \item[(a)] The dimension $p$ satisfies $p=O(n^{\gamma_0})$, and the number of machines $m$ satisfies $\sqrt{\tilde{s}\log(1/\delta)}\log p/\epsilon =o(m)$.  We take
        \begin{equation}    \label{eq:mean_lambdaN}
            \lambda_N\geq C_1\Big(\sqrt{\frac{\log p}{N}}+\frac{\sqrt{\tilde{s}\log(1/\delta)}\log p}{m\sqrt{n}\epsilon}+\frac{1}{n}\Big);
        \end{equation}
        \item[(b)] Define $S=\supp(\vect{\theta}^*)$, and assume that 
        \begin{equation}    \label{eq:mean_mu_min}
            \min_{l\in S}|\theta_l^*|\geq C_2\Big(\sqrt{\frac{\log p}{N}}+\frac{\sqrt{\tilde{s}\log(1/\delta)}\log p}{m\sqrt{n}\epsilon}+\frac{1}{n}\Big).
        \end{equation}
    \end{itemize}
    Then for $\tilde{s}\geq s,\gamma_1>0$ and $\hat{\vect{Q}}(\mbX)$ defined in Algorithm \ref{alg:sign_mean}, we have
    \begin{equation*} 
        \lim_{N\rightarrow\infty}\mbP\Big(\hat{\vect{Q}}(\mbX) = \sgn(\vect{\theta}^*)\Big)=1.
    \end{equation*}
\end{theorem}

As can be seen from assumptions (a) and (b), there are three terms in \eqref{eq:mean_lambdaN} and $\eqref{eq:mean_mu_min}$. The first term corresponds to the optimal signal-to-noise ratio, the second term is incurred by the privacy constraint, and the last term is caused by the error in the Berry-Esseen bound. 

From Theorem \ref{thm:medsign_dp}, $\dpvote$ is $(\epsilon,\delta)$-DGDP. We now comment on the lower bound condition in (\ref{eq:mean_mu_min}). First, we fix the privacy level $(\epsilon,\delta)$. Assume that the number of local machines $m=O(n)$ and take $\tilde{s}$ in the Peeling algorithm satisfying $\sqrt{\tilde{s}\log p\log(1/\delta)}/\epsilon = O(\sqrt{m})$. Thus, (\ref{eq:mean_mu_min}) becomes $\min_{l\in S}|\theta_l^*|\geq C\sqrt{\log p/N}$. This is the classical minimax rate-optimal lower bound for variable selection. That is, if we replace $C\sqrt{\log p/N}$ with $o(\sqrt{\log p/N})$, then for any statistics $T_{N}$ constructed from $\{\vect{X}_1,\dots,\vect{X}_N\}$, there is a sufficiently small constant $c>0$ and a parameter $\vect{\theta}^*$ satisfying
$\min_{l\in S}|\theta_l^*|=c\sqrt{\log p/N}$ such that the support of $\vect{\theta}^{*}$ can not be estimated by $T_{N}$ consistently. See Theorem 3 in \cite{cai_liu_xia.2014jrssb} for further details. 

At the end of this section, we remark that the private sign selection of the mean vector has a wide range of applications. Because we allow each coordinate of the vector $\vect{X}$ to be correlated, it is not difficult to extend our method to differentially private support recovery for the sparse covariance matrix \citep{karoui.2008aos, bickel_levina.2008, cai_liu.2011} and the sparse precision matrix \citep{yuan.2010jmlr, cai_liu_luo.2011}.


\section{Private Sign Selection of Linear Regression} \label{sec:lasso_recov}

In this section, we leverage the idea of $\dpvote$ for the sign selection problem for sparse linear regression in distributed systems. With some commonly adopted assumptions on the covariate vector and the noise, we prove the effectiveness and privacy of our method.


\subsection{$\dpvote$ for Regression Parameter}

Let $(\vect{X}_i,Y_i)$ (where $i=1,\dots,N$) be the i.i.d. observations from the model
\begin{equation}    \label{eq:lin_model}
	Y=\vect{X}^{\tp}\vect{\theta}^*+z,
\end{equation}
where $\vect{\theta}^*=(\theta_1^*,\dots,\theta_p^*)^{\tp}$ is the true sparse regression parameter, and $z$ is the noise independent of the covariate $\vect{X}$. The full dataset is denoted by $\mbX=\{(\vect{X}_1,Y_1),\dots,(\vect{X}_N,Y_N)\}$, and it is assumed that $\mbX$ is divided into $m$ local machines $\mcH_j$ ($1\leq j\leq m$), each holds $n_j$ samples. Similarly, we attempt to recover the sign vector $\sgn(\vect{\theta}^*)=(\sgn(\theta^*_1),\dots,\sgn(\theta^*_p))^{\tp}$.

By leveraging the idea of $\dpvote$, we present Algorithm \ref{alg:priv_reg} for the sign recovery of sparse regression.
Let
	\begin{equation}    \label{eq:local_Lasso}
				\hat{\vect{\theta}}_j(\lambda)=\argmin{\vect{\theta}\in\mbR^p}\frac{1}{2n_j}\sum_{i\in\mcH_j}(Y_i-\vect{X}_i^{\tp}\vect{\theta})^2+\lambda|\vect{\theta}|_1.
			\end{equation}
We use sgn$(\hat{\vect{\theta}}_j(\lambda))$ for Majority Vote. We need to choose a problem-specific regularization parameter $\lambda_j$ for each local machine. In particular, we let $\lambda_j$ be the smallest number such that the local estimator $\hat{\vect{\theta}}_j(\lambda_j)$ having at most $\tilde{s}$ nonzero elements and no less than a universal constant $\lambda_N$ in (\ref{eq:regression_lambdaN}). More precisely, we let
\begin{equation}    \label{eq:lambda_set}
    \lambda_j = \min\left\{\lambda\;\middle|\;|\hat{\vect{\theta}}_j(\lambda)|_{0}\leq\tilde{s}, ~\lambda\geq\lambda_N\right\}.    
\end{equation}
It is worth noting that $\lambda_j$'s could be different for every local machine. It should also be noted that the $\lambda_j$'s always exist. To demonstrate this, we observe that the resulting estimator $\hat{\vect{\theta}}_j(\lambda)$ is $\vect{0}$ if $\lambda$ is sufficiently large. This implies that the set on the right-hand side of \eqref{eq:lambda_set} is always nonempty. For the implementation of (\ref{eq:local_Lasso}), we can use the LARS algorithm \citep{efron_hastie_etal.2004aos} because it obtains all the solution paths efficiently. Hence the choice of $\lambda_{j}$ can be quickly implemented.

\begin{algorithm}[!t]
	\caption{{\small Differentially private Majority Vote for sparse linear regression ($\dpvote$ Lasso)}}
	\label{alg:priv_reg}
	\hspace*{\algorithmicindent} \hspace{-0.5cm}   {\textbf{Input:} Data on local machines $\{(\vect{X}_i,Y_i)| i\in\mcH_j\}$ for $j=1,\dots,m$, the universal regularization parameter $\lambda_N$ in \eqref{eq:regression_lambdaN} and $\lambda_{j}$ in (\ref{eq:lambda_set}), privacy level $(\epsilon,\delta)$.} 	
	\begin{algorithmic}[1]
		\FOR{$j=1,\dots,m$}
		\STATE Let $\hat{\vect{\theta}}_j=\hat{\vect{\theta}}_j(\lambda_{j})$ in (\ref{eq:local_Lasso}),
			and obtain the sign vector $\vect{Q}_j = \text{sgn}(\hat{\vect{\theta}}_j)$. Send $\vect{Q}_{j}$ to the server.
		\ENDFOR
		\STATE Apply $\dpvote$ 
		\begin{equation}	\label{eq:dpmed-lasso_sign}
			\hat{\vect{Q}}(\mbX)=\dpvote(\{\vect{Q}_j\}_{1\leq j\leq m},\tilde{s},\epsilon,\delta).
		\end{equation}
	\end{algorithmic}
	 \textbf{Output:}  The sign vector $\hat{\vect{Q}}(\mbX)$.
\end{algorithm}


\subsection{Theory of Regression Parameter Sign Recovery}	\label{sec:priv_lasso_theory}

For linear regression, we consider the following distribution space
\begin{equation}	\label{eq:cov_dist}
\begin{aligned}
	\mcP_{\vect{X}, Y}(\vect{\theta}^*, \rho, \eta_1, C_1,\eta_2, C_2)=\Big\{\mbP\Big|& \rho\leq\Lambda_{\min}(\vect{\Sigma})\leq\Lambda_{\max}(\vect{\Sigma})\leq \rho^{-1}, \\ & \sup_{|\vect{v}|_2=1}\mbE_{\mbP}\big\{\exp(\eta_1|\vect{v}^{\tp}\vect{X}|^2)\big\}\leq C_1, \\ 
	&z=Y-\vect{X}^{\tp}\vect{\theta}^*,z\perp\vect{X}, \mbE_{\mbP}\big\{\exp(\eta_2|z|^2)\big\}\leq C_2  \Big\},
\end{aligned}
\end{equation}
where $\rho,C_1,C_2,\eta_1,\eta_2$ are positive constants. This implies that both the covariate vector $\vect{X}$ and the noise $z$ are sub-Gaussian. Then, we obtain the following sign consistency result.

\begin{theorem} \label{prop:dp_lasso_sign}
    (Sign consistency of $\dpvote$ Lasso) Let $\mbX=\{(\vect{X}_1,Y_1),\dots,(\vect{X}_N,Y_N)\}$ be $N$ i.i.d. random vectors sampled from $\mcP_{\vect{X},Y}(\vect{\theta}^*,\rho,\eta_1,C_1,\eta_2,C_2)$ distributed in $m$ machines $\mcH_1,\dots,\mcH_m$, with each machine holds $n_j$ ($1\leq j\leq m$) samples. Here we assume $n_1\asymp ...\asymp n_m\asymp n$. Moreover, assume that sufficiently large constants $C_3,C_4,\gamma_0$ exist such that
    \begin{itemize}
		\item[(a)] The dimension, the number of machines and the sparsity level satisfy $p=O(n^{\gamma_0})$,\\$\sqrt{\tilde{s}\log(1/\delta)}\log p/\epsilon =o(m), \tilde{s}=o(\sqrt{n/\log p})$, and we take 
		\begin{equation}    \label{eq:regression_lambdaN} 
		\lambda_N=C_3\Big(\sqrt{\frac{\log p}{N}}+\frac{\sqrt{\tilde{s}\log(1/\delta)}\log p}{m\sqrt{n}\epsilon}+\frac{1}{n}\Big);
		\end{equation}
		\item[(b)] For $S=\supp(\vect{\theta}^*)$, the minimal signal satisfies 
		\begin{equation}    \label{eq:regression_theta_min}
		    \min_{l\in S}|\theta_l^*|\geq C_4\Big(\sqrt{\frac{\log p}{N}}+\frac{\sqrt{\tilde{s}\log(1/\delta)}\log p}{m\sqrt{n}\epsilon}+\frac{1}{n}+\max_{1\leq j\leq m}\lambda_j\Big),
		\end{equation}
		with probability tending to $1$. 
		\item[(c)] The covariance matrix $\vect{\Sigma}=\mbE\vect{X}\vect{X}^{\tp}$ is positive definite. Let $\vect{\Sigma}^{-1}=(\vect{\omega}_1,\dots,\vect{\omega}_p)$, assume that
		\begin{equation}\label{eq:incoherence}
			\max_{l\in S^c}\frac{|\vect{\omega}_{-l}|_1}{\omega_{l,l}}\leq 1-\Delta_0.
		\end{equation}
	\end{itemize}
    Then for $\tilde{s}\geq s$, the estimator $\hat{\vect{Q}}(\mbX)$ defined in Algorithm \ref{alg:priv_reg} satisfies 
    \begin{equation*}
    	\lim_{N\rightarrow\infty}\mbP\Big(\hat{\vect{Q}}(\mbX) = \sgn(\vect{\theta}^*)\Big)= 1.
    \end{equation*}
\end{theorem}

From Theorem \ref{thm:medsign_dp} we know $\dpvote$ regression is $(\epsilon,\delta)$-DGDP. We now make some comments regarding these assumptions. As can be seen from assumption (a), the lower bound of the regularization parameter $\lambda_N$ has three terms, which is the same as in \eqref{eq:mean_lambdaN}. In assumption (b), we assume $\eqref{eq:regression_theta_min}$ holds with a probability tending to $1$. Here, randomness comes from the selection of $\lambda_j$'s. Note that each $\lambda_j$ is chosen according to \eqref{eq:lambda_set} which depends on the data.  Assumption (c) implies that, for $l\in S^c$, the $l$-th row of the precision machine $\vect{\Sigma}^{-1}$ is dominated by the diagonal entry $\omega_{l,l}$. 

\begin{remark}
 In the lower bound condition \eqref{eq:regression_theta_min} on the signals,  there is an additional term $\max_{1\leq j\leq m}\lambda_j$, which depends on the magnitude of each $\lambda_j$. From the classical sparsity result for Lasso, under some regular conditions, $\supp(\hat{\vect{\theta}}_j)\subseteq S$ with probability tending to one, when $\lambda_{j}\geq \tilde{C}\sqrt{\log p/n}$ for a large $\tilde{C}>0$ (see \cite{wainwright2009sharp}). Therefore, it can be seen see that $\lambda_{j}$ in (\ref{eq:lambda_set}) lies in the interval $(\lambda_N,\tilde{C}\sqrt{\log p/n})$ with a probability tending to one.  Furthermore,
if the dimension $p$ satisfies $p=o(\sqrt{n/\log n})$ and $\sqrt{p\log p\log(1/\delta)}/\epsilon = O(\sqrt{m})=O(\sqrt{n})$, we can simply set $\tilde{s}=p$, then,
we have $\lambda_{j}=\lambda_{N}$. In this case, $\lambda_{N}\approx C\sqrt{\log p/N}$ and condition (\ref{eq:regression_theta_min}) becomes $\min_{l\in S}|\theta_l^*|\geq C\sqrt{\log p/N}$ for some $C>0$. This lower bound is optimal for variable selection consistency.  
\end{remark}


\section{Simulation Study}	\label{sec:sim}

In this section, we conduct two experiments to examine the performance of our Majority Vote procedure and the relevant differentially private versions. Since our study is the first to examine DGDP in a distributed system, it
is difficult to compare it with other methods. However,  we still choose the methods of \cite{cai_wang_zhang.2019}
as the object for comparison. There are two main reasons for this choice. First, the approaches suggested by \cite{cai_wang_zhang.2019} are based on iterative algorithms and can be implemented in a distributed system. Second, from comparisons with \cite{cai_wang_zhang.2019}, we can see that the methods designed for parameter estimation perform poorly for sign/variable selection. Therefore, it is important to investigate the problem of sign/variable selection  separately, as we claimed in the introduction. As we pointed out in the section on related works, there are also some special works on the variable selection with DP. However, some of them are computationally inefficient, and others cannot be applied directly to the distributed setting.

\subsection{Results for Sparse Mean Estimation}
In the first experiment, we consider the sparse mean estimation problem. The observations $\{\vect{X}_1,\dots,\vect{X}_N\}$ are sampled from the model
\begin{equation*}
	\vect{X}_i=\vect{\theta}^*+\vect{z}_i,
\end{equation*} 
where the noises $\vect{z}_i$'s are drawn from a multivariate normal distribution $\mcN(\vect{0},\vect{\Sigma})$. Here, the covariance matrix $\vect{\Sigma}$ is a $p\times p$ Toeplitz matrix with its $(i,j)$-th entry $\Sigma_{ij}=0.5^{|i-j|}$, where $1\leq i,j\leq p$. We set the dimension $p$ to $500$. The parameter of interest is defined as:
\begin{equation}    \label{eq:sim_mu}
	\vect{\theta}^*=(\,1,0.8,\cdots,0.2,-0.2,\cdots,-0.8,-1,\vect{0}^{\tp}_{p-10}\,)^{\rm T},
\end{equation}
which implies that the sparsity level $s$ is fixed at $10$. For the regularization parameter $\lambda_N$ on each local machine, we take $\lambda_N=0.1$ for simplicity. 
For the privacy level $(\epsilon,\delta)$, we take $\epsilon=0.5$ and $\delta = 0.05$. 
Next, we check the performance of the following three methods: Majority Vote without adding noise, $\mathrm{DPvote}$ in Algorithm \ref{alg:sign_mean}, and the method by \cite{cai_wang_zhang.2019}. 
\begin{enumerate}[label=(\alph*)]
	\item $\mathrm{Vote}$: Majority Vote without adding noise;	
	\item $\mathrm{DPvote}$: Majority Vote with DGDP;
	\item $\mathrm{CWZ}$: The method proposed in \cite{cai_wang_zhang.2019} for all samples on a single machine. Here we take the truncation level $R$  in \cite{cai_wang_zhang.2019} to be $2$.	
\end{enumerate}

Note that $\mathrm{Vote}$ is a non-privacy algorithm. For $\mathrm{DPvote}$ and $\mathrm{CWZ}$, we set the selection sparsity level $\tilde{s}$ to $15$.  The original $\mathrm{CWZ}$ is designed to protect the privacy of a single element. To protect the privacy of the whole local samples, we multiply the noise level with the local sample size $n$ by the composition theorem in Lemma \ref{lem:two_comp}. The performance of sign recovery is measured using the following two criteria:
\begin{itemize}	
	\item \textbf{False Discovery Rate (FDR). }
		\begin{equation*}
			\mathrm{FDR}=\frac{\sum_{l\notin S}\mbI\{\hat{Q}_l(\mbX)\neq0\}+\sum_{l\in S^-}\mbI\{\hat{Q}_l(\mbX)=1\}+\sum_{l\in S^+}\mbI\{\hat{Q}_l(\mbX)=-1\}}{\max[\sum_{l=1}^p\mbI\{\hat{Q}_l(\mbX)\neq0\},1]}.
		\end{equation*}
	\item \textbf{Power. }
		\begin{equation*}
			\mathrm{Power}=\frac{\sum_{l\in S}\mbI\{\hat{Q}_l(\mbX)=\sgn(\theta^*_l)\}}{\max[|S|,1]}.
		\end{equation*}
\end{itemize}

\paragraph{Effect of the number of machines} To evaluate the effect of the number of machines $m$, we fix the local sample size $n$ to be $500$ and vary the number of machines $m$ from $500$ to $1500$ by $100$. The total sample size $N=nm$ varies accordingly. Then we report the $\mathrm{FDR}$ and $\mathrm{Power}$ of the three methods: $\mathrm{DPvote}$, $\mathrm{Vote}$,  and $\mathrm{CWZ}$. The results are shown in Figure \ref{fig:mean_machine}.

As we can see, $\vote$ always performs the best on both the power and FDR. This is not surprising, as we do not add any noise to $\vote$. The $\mathrm{Power}$ of $\dpvote$ is also close to one and is higher than that of $\mathrm{CWZ}$ significantly. On the other hand, the $\mathrm{FDR}$s of $\vote$ and $\dpvote$ are close to zero, whereas the FDRs of $\mathrm{CWZ}$ are quite large. The behavior of $\mathrm{CWZ}$ indicates that procedures work well in parameter estimation may be not suitable for sign selection.

\paragraph{Effect of the privacy budget} To study the effect of the privacy budget $\epsilon$, we fix the number of machines $m$ to be $800$, the local sample size $n$ to be $500$, and vary the privacy level $\epsilon$ from $0.1$ to $1$ by $0.1$. Because $\vote$ is a non-private algorithm, we focus on the other two private algorithms in this part: $\mathrm{DPvote}$ and $\mathrm{CWZ}$. The $\mathrm{FDR}$ and $\mathrm{Power}$ of these methods are presented in Figure \ref{fig:mean_privacy}.

With an increase in the privacy level $\epsilon$, the performance of $\mathrm{DPvote}$ improves. This is because a larger privacy level $\epsilon$ induces smaller random noise in private algorithms. Note that $\mathrm{DPvote}$ outperforms $\mathrm{CWZ}$ uniformly at all  privacy levels. Although $\mathrm{CWZ}$ behaves better as $\epsilon$ increases, it still has a large FDR even when $\epsilon=1$.

\begin{figure}[!h]
	\begin{center}
		\includegraphics[width=0.8\textwidth]{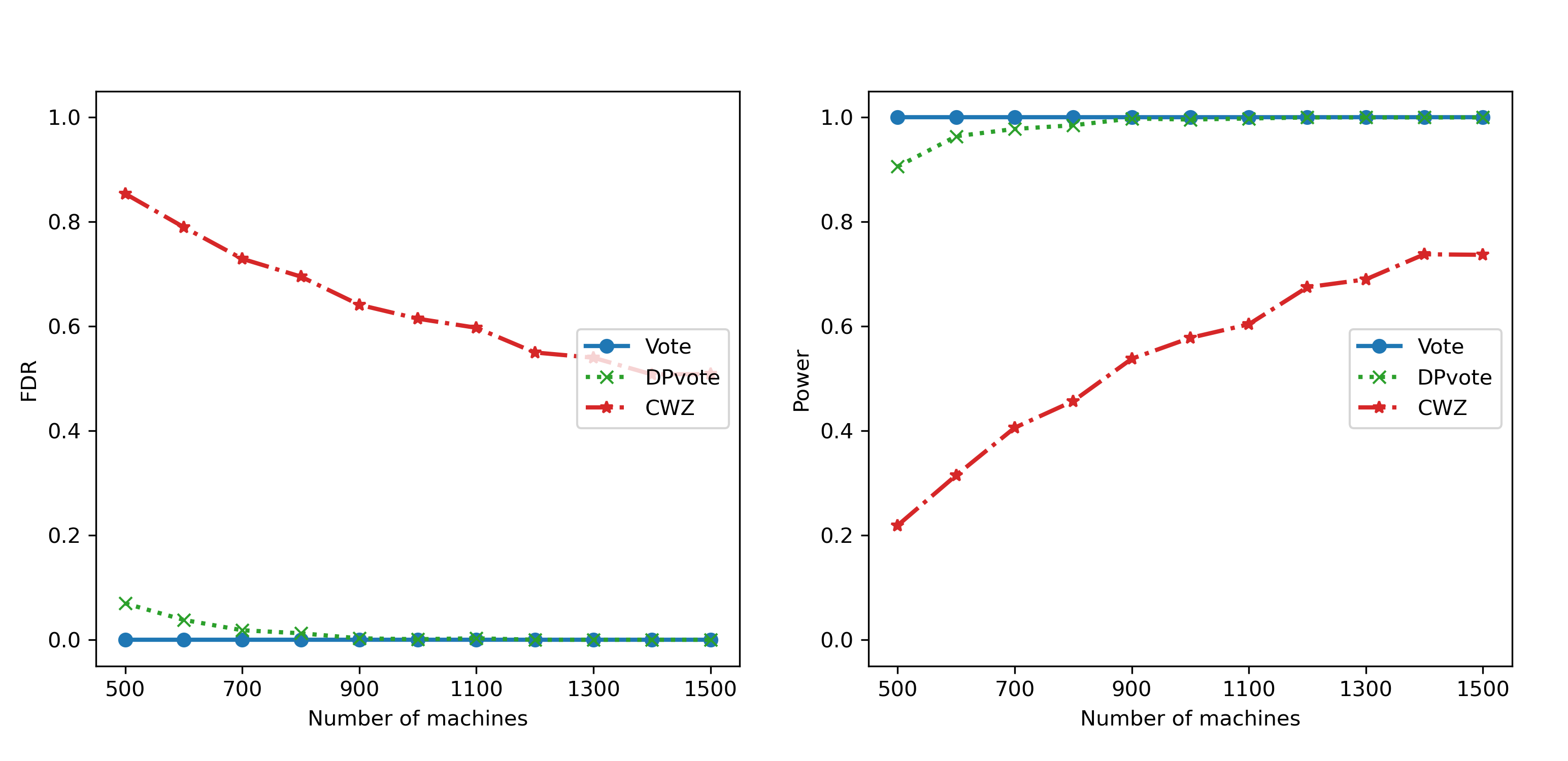}
	\end{center}
	\caption{The FDR and Power over the number of machines for sparse mean estimation. The number of machines varies from $500$ to $1500$, the local sample size is $500$, and the dimension $p$ is $500$.} 
	\label{fig:mean_machine}
\end{figure}

\begin{figure}[!h]
	\begin{center}
		\includegraphics[width=0.8\textwidth]{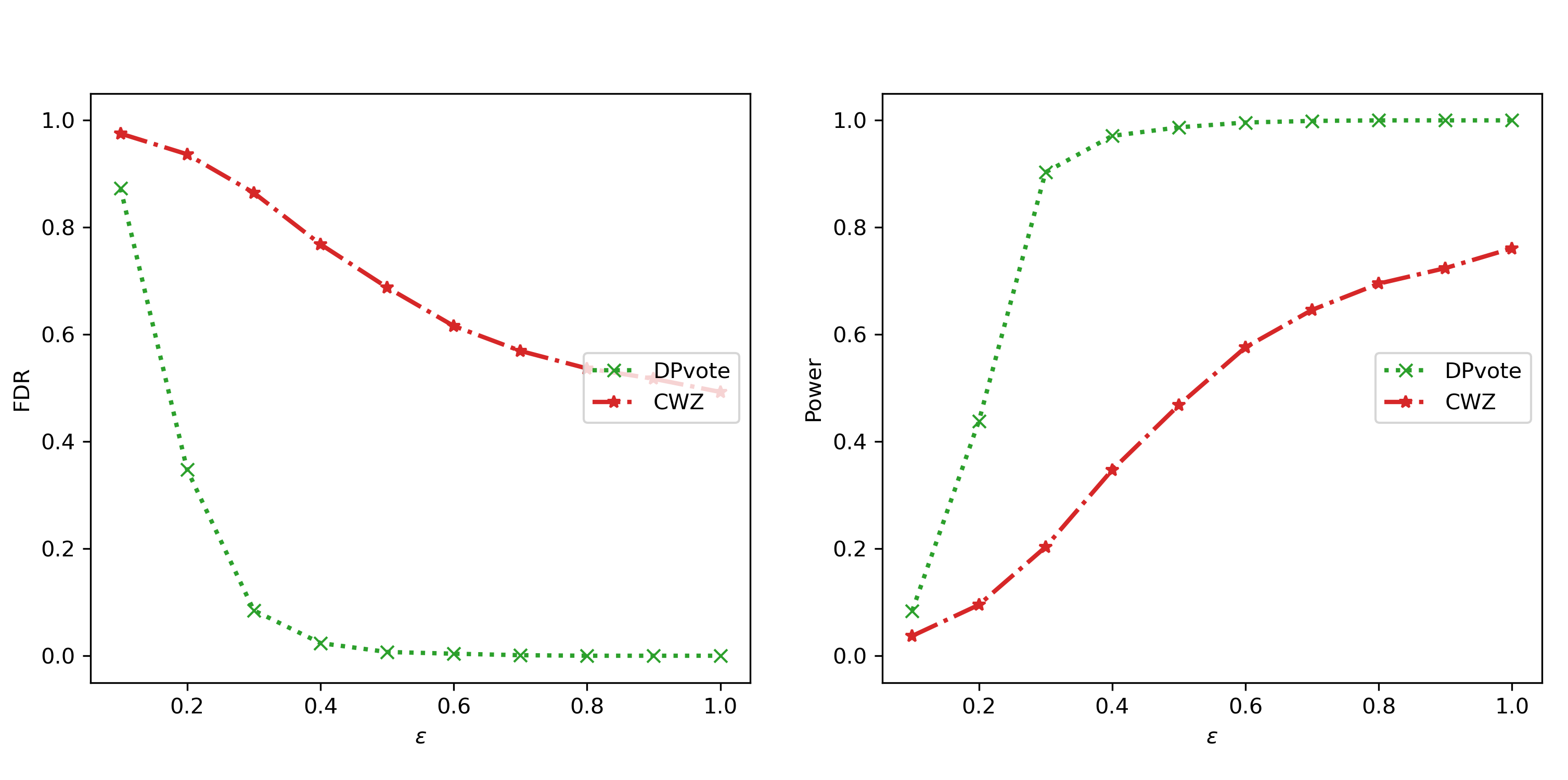}
	\end{center}
	\caption{The FDR and Power over the privacy level $\epsilon$ for sparse mean estimation. The local sample size is $500$, the number of machines is $800$, and the dimension $p$ is $500$.}
	\label{fig:mean_privacy}
\end{figure}

\subsection{Results for Sparse Linear Regression}
In this section, we consider the linear model defined by \eqref{eq:lin_model}. The i.i.d. noises $z_i$'s are drawn from $\mcN(0,1)$ and the i.i.d. covariate vectors $\vect{X}_i=( X_{i,1},\dots,X_{i,p})^{\tp}$ ($i=1,\dots,N$) are drawn from a multivariate normal distribution $\mcN(\vect{0},\vect{\Sigma})$. The covariance matrix $\vect{\Sigma}$ is a $p\times p$ Toeplitz matrix with the $(i,j)$-th entry $\Sigma_{ij}=0.5^{|i-j|}$, where $1\leq i,j\leq p$. We fix the dimension $p=200$. Moreover, we set the true coefficient $\vect{\theta}^*$ as the same as $\vect{\theta}^*$ in \eqref{eq:sim_mu}. In the CWZ method proposed in \cite{cai_wang_zhang.2019}, we take the truncation level $R$ to be $2$, the number of steps $T$ to be 20, and the step size $\eta$ to be $0.1$. 

 We adopt the same settings as the counterpart in the experiment for the sparse mean estimation. 
The $\mathrm{FDR}$ and $\mathrm{Power}$  are shown in Figures \ref{fig:lasso_machine} and \ref{fig:lasso_privacy}. It is not surprising that all these methods perform similarly to mean estimation. That is, with an increase in the number of machines $m$ or the privacy level $\epsilon$, the $\mathrm{FDR}$ of $\dpvote$ approaches zero and $\mathrm{Power}$ approaches one. Again, $\dpvote$ outperforms $\mathrm{CWZ}$ in all settings. 

\begin{figure}[!h]
	\begin{center}
		\includegraphics[width=0.8\textwidth]{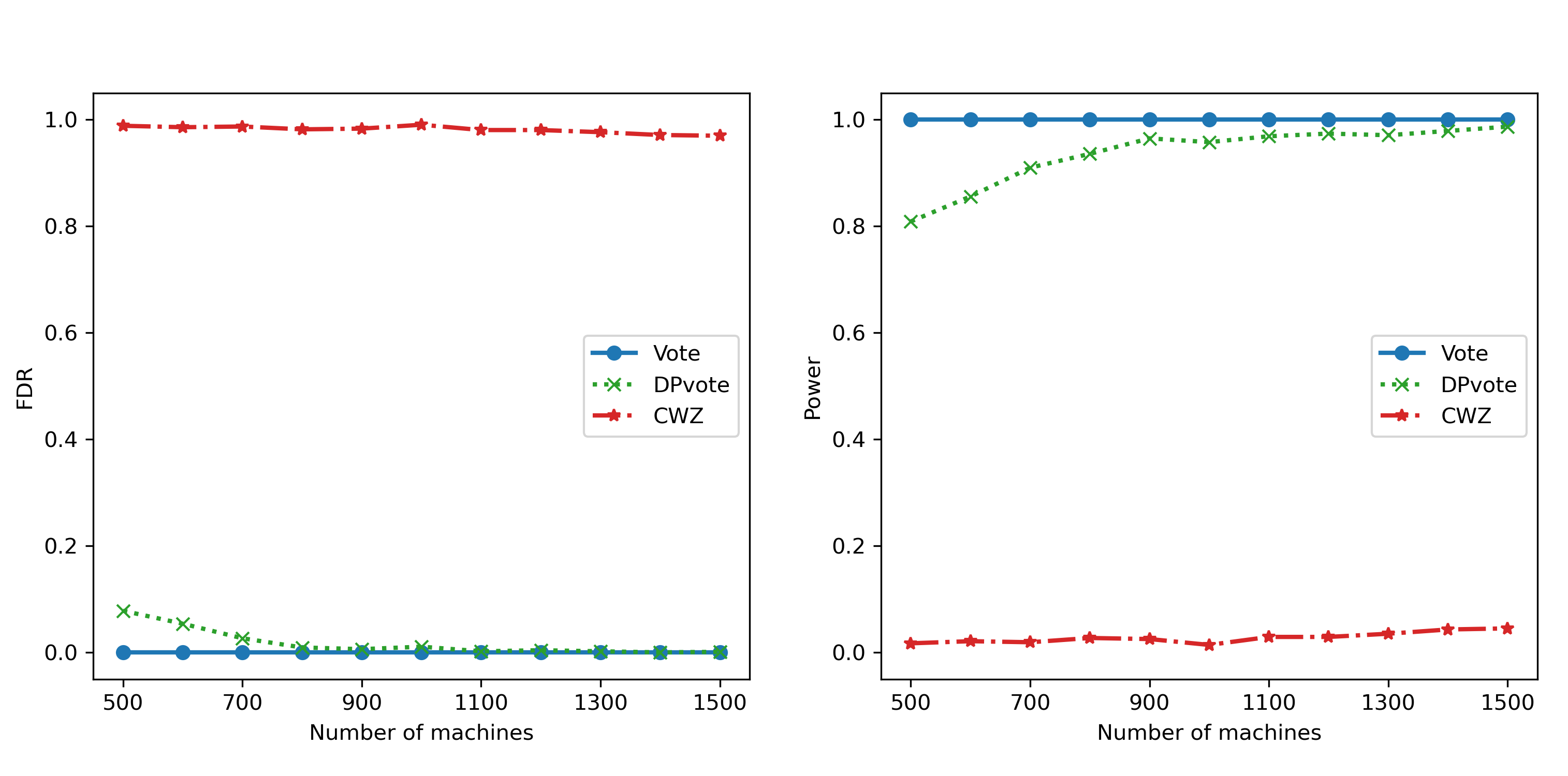}
	\end{center}
	\caption{The FDR and Power over the number of machines for sparse linear regression. The number of machines varies from $500$ to $1500$, the local sample size is $500$, and the dimension $p$ is $500$.}
	\label{fig:lasso_machine}
\end{figure}

\begin{figure}[!h]
	\begin{center}
		\includegraphics[width=0.8\textwidth]{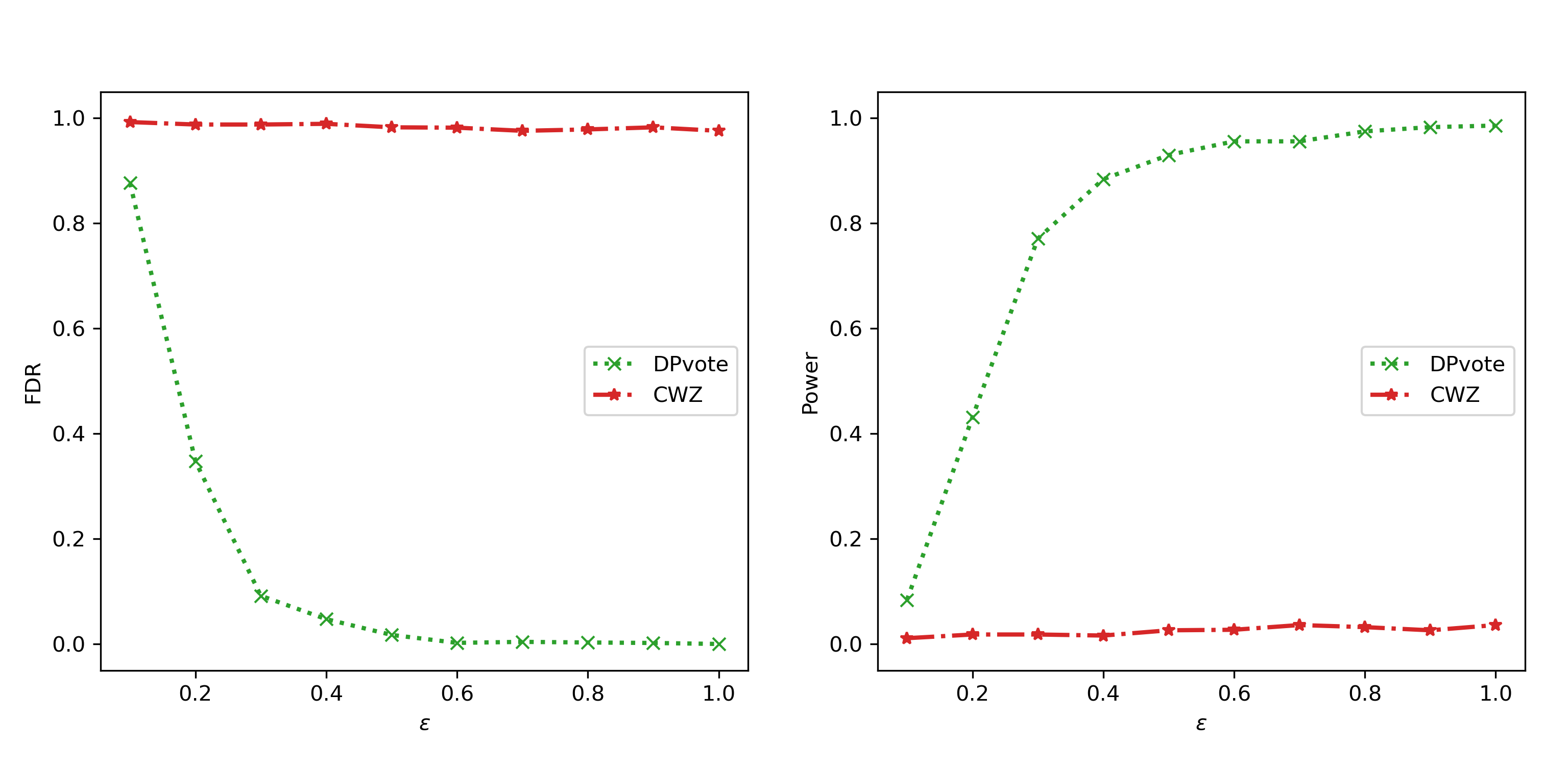}
	\end{center}
	\caption{The FDR and Power over the privacy level $\epsilon$ for sparse linear regression. The local sample size is $500$, the number of machines is $800$, and the dimension $p$ is $500$.}
	\label{fig:lasso_privacy}
\end{figure}

\section{Conclusion}	\label{sec:conclude}

In summary, this study considers the sign selection problem under privacy constraints and applies it to private variable selection for sparse mean vectors and sparse linear regression in a distributed setup. Using the differentially private Majority Vote algorithm, we can consistently recover the signs. The admissible signal-to-noise ratio can have the same order as that of the non-private algorithm. In the following, we will discuss several additional topics and intriguing avenues for future research.

In the following section, we will explore several additional topics and intriguing avenues for future research.

\begin{itemize}
	\item \textbf{More sparse recovery problems.} In this paper, we apply the DPvote algorithm to variable selection for sparse mean vectors and sparse linear regression. It is natural to extend our method to differentially private support recovery for sparse covariance matrix \citep{karoui.2008aos, bickel_levina.2008, cai_liu.2011} and sparse precision matrix \citep{yuan.2010jmlr, cai_liu_luo.2011}. It is also interesting to consider applying our method to the debias-Lasso problem \citep{javanmard2014confidence, lee_liu_etal.2017, javanmard_montanari.2018};
	\item \textbf{Untrusted curator.} In the DGDP definition, we make the assumption that a trusted curator is present to collect information from each machine and execute the algorithm. In scenarios where the curator cannot be trusted, secure multi-party computation (SMC) techniques are typically employed to prevent any party from learning intermediate information \citep{goldreich1998secure, du_atallah.2001, evans_etal.2018}. While this setting is beyond the scope of our current paper, It would be interesting to explore combining our method with SMC techniques to address this issue. Moreover, this scenario bears a resemblance to the local DP setup, where additional noise is added when transmitting data to an untrusted curator \citep{duchi_jordan_wainwright.2018}. It would be possible to explore combining techniques from local DP with our method to tackle this privacy concern;
	\item \textbf{Optimal minimal signal condition.} In condition (b) of Theorem 3, we can find that the minimal signal condition is lower bounded by three terms. The second term $\sqrt{\tilde{s}\log(1/\delta)}$\mbox{}\\$\log p/(m\sqrt{n}\epsilon)$ is new and arises from the privacy constraint. In \cite{butucea_etal.2020arXiv}, the author provided a $d/(\epsilon\sqrt{n})$  lower bound of the minimal signal under local DP setup (see Corollary 2.2). We acknowledge that the proof techniques for DP algorithms and local DP algorithms are different (see \cite{cai_wang_zhang.2019} and \cite{duchi_jordan_wainwright.2018}). Additionally, in \cite{bafna_ullman.2017colt} and \cite{steinke_ullman.2017focs}, the authors demonstrated sample complexity lower bounds for DP selection, independent of the privacy levels $\epsilon$ and $\delta$.  It would be interesting to develop a minimax theory for minimal signal under DP, as well as our proposed DGDP setups. 
\end{itemize}

\bibliographystyle{asa}
\bibliography{private_med_dc}

\newpage

\section{Appendix}



\subsection{Proof of Results in Section \ref{sec:medsign_theory}}\label{sec:proof_2-4}

\begin{proof}[Proof of Theorem \ref{thm:medsign_dp}]
	Notice that the stability function $f^S(\cdot)$ has sensitive $2$. By Lemma 2.4 of \cite{dwork_su_li.2018differentially}, we know each round of selection is $(\epsilon',0)$ private, where
	\begin{equation*}
		\epsilon' = \frac{\epsilon}{4\sqrt{2\tilde{s}\log(2/\delta)}}.
	\end{equation*}
	Then by advanced composition theorem (Lemma \ref{thm:adv_comp}), we know that the selection of the set $\tilde{S}$ with $\tilde{s}$ indices is $(\frac{\epsilon}{2},\frac{\delta}{2})$-differentially private. Notice that the functions $f_l(\cdot,1),f_l(\cdot,-1),f_l(\cdot,0)$ all have global sensitivity $2$. By Theorem 3.10 of \cite{dwork_aaron.2014}, for every $l\in\tilde{S}$, the generation of $\hat{Q}_l$ is $(\epsilon',0)$-differentially private. Then by the advanced composition theorem (Lemma \ref{thm:adv_comp}), we know that the generation of the $\tilde{s}$-tuple $\{\hat{Q}_l|l\in \tilde{S}\}$ is $(\frac{\epsilon}{2},\frac{\delta}{2})$-differentially private. Next, we apply the composition theorem (Lemma \ref{lem:two_comp}) and conclude that the composition of these two mechanisms is $(\epsilon,\delta)$-differentially private.
\end{proof}

\begin{proof}[Proof of Theorem \ref{thm:medsign_cons}]\mbox{} To prove the selection consistency of $\dpvote$, we first prove the consistency of $\peel$.

	\noindent\textbf{Consistency of $\peel$:} Let us directly consider the extreme case, where $f^S(\vect{Q}_l^r)$\\ 
	$=\lceil 8\gamma\sqrt{2\tilde{s}\log(2/\delta)}\log p\rceil$ for $l\in \bar{S}$, and $f^S(\vect{Q}_l^r) = 0$ for $l\notin \bar{S}$. Clearly the probability of $\dpvote$ algorithm produces the correct signs for any dataset $\mbQ$ is not less than this case. For simplicity denote it as $\mbP(\text{extreme})$. 
	
	On the other hand, for the $t$-th selection step (where $1\leq t\leq \tilde{s}$), suppose there left $\bar{s}-t+1$ elements in $\bar{S}$ and $p-\bar{s}$ elements in $\bar{S}^c$, the probability that the peeling algorithm select element in $\bar{S}$ is greater than the case when there left only one element in $\bar{S}$, let's denote this probability as $\mbP(1\text{ vs }p-\bar{s})$. Since each round of selection is independent, we know
	\begin{equation}	\label{eq:cons_bound1}
		\mbP(\text{extreme}) \geq \mbP^{\bar{s}}(1\text{ vs }p-\bar{s}).
	\end{equation}
	Therefore, it left to bound the probability $\mbP(1\text{ vs }p-\bar{s})$. 
	
	By Lemma \ref{lem:Lap_max}, we know that
	\begin{align*}
		\mbP(\max_{l\notin \bar{S}}f^S(\vect{Q}_l^r)+g_l\geq 0)
		\leq\mbP(\max_{l\notin \bar{S}}g_l\geq 8\gamma\sqrt{2\tilde{s}\log(2/\delta)}\log p)\leq p^{-\gamma+1},\\
		\mbP(\max_{l\in \bar{S}}f^S(\vect{Q}_l^r)+g_l\leq 0)
		\leq\mbP(\max_{l\in \bar{S}}g_l\leq -8\gamma\sqrt{2\tilde{s}\log(2/\delta)}\log p)\leq  p^{-\gamma+1}.
	\end{align*}
	Therefore, 
	\begin{align*}
		\mbP(1\text{ vs }p-\bar{s}) \geq 1 - \mbP(\max_{l\notin \bar{S}}f^S(\vect{Q}_l^r)+g_l\geq 0)-\mbP(\max_{l\in \bar{S}}f^S(\vect{Q}_l^r)+g_l\leq 0)\geq 1-2p^{-\gamma+1}.
	\end{align*}
	Plugging it into \eqref{eq:cons_bound1} we have that
	\begin{equation}	\label{eq:cons_keybound1}
		\mbP(\text{extreme}) \geq \mbP^{\bar{s}}(1\text{ vs }p-\bar{s})\geq (1-2p^{-\gamma+1})^{\bar{s}}\geq 1-2\bar{s}p^{-\gamma+1}\geq 1-O(p^{-\gamma+2}),
	\end{equation}
	which proves the first part.
	
	\noindent\textbf{Consistency of $\dpvote$:} Suppose that the $\peel$ algorithm has selected the support covering $\bar{S}$, i.e. $\bar{S}\subseteq\tilde{S}$. For each $l\in \bar{S}$, without loss of generality assume $\bar{Q}_l=1$. Then from assumption we know 
	\begin{equation*}
		f^S(\vect{Q}_l^r) = N_l^+- N_l^0-N_l^-   \geq 8(\gamma+1)\sqrt{2\tilde{s}\log(2/\delta)}\log p/\epsilon.
	\end{equation*}
	Then we have that
	\begin{align*}
		\mbP(\hat{Q}_l=1) =& \frac{P_l^+}{P_l^++P_l^-+P_l^0}\\
		=&\frac{1}{1+ \exp\{\epsilon' (N^0_l-N^+_l+N_l^-)/2\}+ \exp\{\epsilon' (N^-_l-N^+_l)/2\}}\\
		\geq&\frac{1}{1+ \exp\{-\epsilon'f_l(\vect{Q}_l^r,1)/2\}+ \exp\{\epsilon' f_l(\vect{Q}_l^r,1)/2\}}\\
		\geq&\frac{1}{1+ 2p^{-\gamma-1}}\geq 1-2p^{-\gamma-1}.
	\end{align*}
	For $l\in\tilde{S}\backslash \bar{S}$,  we have that
	\begin{align*}
		\mbP(\hat{Q}_l=0) =& \frac{P_l^0}{P_l^++P_l^-+P_l^0}\\
		=&\frac{1}{1+ \exp\{-\epsilon' \min\{N^0_l-N^+_l+N^-_l,N_l^0\}/2\}+ \exp\{\epsilon' \min\{N^0_l-N^-_l+N^+_l,N_l^0\}/2\}}\\
		\geq&\frac{1}{1+ 2\exp\{-\epsilon' \min\{N^0_l-N^+_l+N^-_l,N^0_l+N^+_l-N^-_l\}/2\}}\\
		\geq&\frac{1}{1+ 2p^{-\gamma-1}}\geq 1-2p^{-\gamma-1}.
	\end{align*}
	Then we have that
	\begin{equation*}
		\mbP(\hat{Q}_l=\bar{Q}_l,l\in\tilde{S})=\prod_{l\in\tilde{S}}\mbP(\hat{Q}_l=\bar{Q}_l)\geq (1-2p^{-\gamma-1})^{\bar{s}}\geq 1-2\bar{s}p^{-\gamma-1}\geq 1-p^{-\gamma}.
	\end{equation*}
	Therefore,
	\begin{align*}
		\mbP(\dpvote\text{ recovers }\bar{\vect{Q}}) \geq \mbP(\hat{Q}_l=\bar{Q}_l,l\in\tilde{S}) - (1-\mbP(\text{extreme}))\geq 1-p^{-\gamma} - p^{-\gamma+2} =1-O(p^{-\gamma+2}),
	\end{align*}
	which proves the Theorem \ref{thm:medsign_cons}.
\end{proof}

\begin{lemma}{(Theorem 1 of \cite{arratia_etal.1989aop})}	\label{lem:possion_approx}
	Let $\{g_{i}|i\in I\}$ be the collection of random variables in an index set $I$. For each $i\in I$, let $I_i$ be a subset of $I$ with $i\in I_i$. For a given $x\in\mbR$, set $y=\sum_{i\in I}\mbP(g_i>x)$. Then
	\begin{equation}	\label{eq:possion_approx}
		\Big|\mbP\Big(\max_{i\in I}g_i\leq x\Big)-e^{-y}\Big|\leq (1\wedge y^{-1})(T_1+T_2+T_3),
	\end{equation}
	where
	\begin{align*}
		T_1 =& \sum_{i\in I}\sum_{j\in I_i}\mbP(g_i>x)\mbP(g_j>x),\\
		T_2 =& \sum_{i\in I}\sum_{j\in I_i,i\neq j}\mbP(g_i>x,g_j>x),\\
		T_3 =& \sum_{i\in I}\mbE\big[|\mbP(g_i>x|\sigma(g_j,j\notin I_i))-\mbP(g_i>x)|\big].
	\end{align*}
\end{lemma}

\begin{lemma}	\label{lem:Lap_max}
	Let $g_1,\dots,g_p$ be $p$ i.i.d. random variable sampled from $\mathrm{Lap}(\lambda)$, then for each $\gamma\geq 2$ we have that
	\begin{equation*}
		\mbP\Big(\max_{i\in I}g_i> \gamma\lambda\log p\Big) \leq p^{-\gamma+1}.
	\end{equation*}
\end{lemma}

\begin{proof}
	To apply Lemma \ref{lem:possion_approx}, we let the index $I=\{1,\dots,p\}$, for every $i$, let $I_i=\{i\}$. Then by taking $x=\gamma\lambda\log p$, we know that
	\begin{align*}
		y &= \sum_{i\in I}\mbP(g_i>x) = \frac{p}{2}e^{-x/\lambda}=\frac{1}{2}p^{-\gamma+1},\\
		T_2& = T_3 = 0,\\
		T_1& = p\sum_{i\in I}\mbP^2(g_i>x)=\frac{p}{4}e^{-2x/\lambda}=\frac{1}{4}p^{-2\gamma+1}.
	\end{align*}
	By plugging them into \eqref{eq:possion_approx} we have that
	\begin{align*}
		&\Big|\mbP\Big(\max_{i\in I}g_i\leq \gamma\lambda\log p\Big)-e^{-\frac{1}{2}p^{-\gamma+1}}\Big|\leq \frac{1}{4}p^{-2\gamma +1},\\
		\Rightarrow&\mbP\Big(\max_{i\in I}g_i> \gamma\lambda\log p\Big) \leq 1-e^{-\frac{1}{2}p^{-\gamma+1}}+\frac{1}{4}p^{-2\gamma+1}\leq p^{-\gamma+1},
	\end{align*}
	which proves the lemma.
\end{proof}


\subsection{Proof of Results in Section \ref{sec:priv_mean_theory}}
\label{sec:proof_3-2}

\begin{lemma}	\label{lemma:berry_esseen}
	(Berry-Esseen Theorem, Theorem 9.1.3 in \cite{chow_teicher.2012}) If $\{X_i, \; i\geq1\}$ are i.i.d. mean-zero random variables with $\mbE[X_i^{2}]=\sigma^2, \mbE|X_i|^{3}<\infty$. Then there exists a constant $C_{\mathrm{B}}>0$ such that
	\begin{equation*}
		\sup_{-\infty<x<\infty}\left|\mbP\left(\sum_{i=1}^nX_i<\sqrt{n}\sigma x\right)-\Phi(x)\right|\leq \frac{C_{\mathrm{B}}}{\sqrt{n}}.
	\end{equation*}
\end{lemma}


\begin{lemma}	\label{lemma:cai_liu}
	(Exponential Inequality, Lemma 1 in \cite{cai_liu.2011}) Let $X_1,...,X_n$ be i.i.d. random variables with zero mean. Suppose that there exist some $\eta>0$ and $C>0$ such that $\mbE [X_i^2e^{\eta|X_i|}]\leq C$. Then uniformly for $0<x\leq C$ and $n\geq 1$, there is
	\begin{equation*}
		\mbP\left\{\frac{1}{n}\sum_{i=1}^nX_i\geq (\eta+\eta^{-1})x\right\}\leq\exp\left(-\frac{nx^2}{C}\right).
	\end{equation*}
\end{lemma}


\begin{lemma}	\label{lemma:concen_mom}
	Let $N$ i.i.d. random variables $X_1,...,X_{N}$ distributed in $m$ subsets $\mcH_1,...,\mcH_m$, with each subset holds $n_j$ ($1\leq j\leq m$) samples. Here we assume $n_1\asymp ...\asymp n_m\asymp n$. Suppose $\mbE[X_i]=0, \mbE[X_i^{2}]=\sigma^2, \mbE|X_i|^{3}<\infty$. Denote $\bar{X}_j=\sum_{i\in\mcH_j}X_i/n_j$ as the local sample mean on $\mcH_j$. For every $\gamma>1$ and non-negative constant $k$, denote
	\begin{equation*}
		c_{\gamma,k}=\tilde{C}\left(\frac{1}{n}+\frac{k}{m\sqrt{n}}+\sqrt{\frac{\gamma\log p}{mn}}\right),
	\end{equation*} 
	where $\tilde{C}>0$ is sufficiently large enough. Then there is
	\begin{equation}	\label{eq:mom_lemma}
		\mbP\left\{\sum_{j=1}^{m}\mbI\left(\bar{X}_j<c_{\gamma,k}\right)\leq\frac{m}{2}+k\right\}+\mbP\left\{\sum_{j=1}^{m}\mbI\left(\bar{X}_j>-c_{\gamma,k}\right)\leq\frac{m}{2}+k\right\}=O(p^{-\gamma}).
	\end{equation}
\end{lemma}

\begin{proof}
	For every $x>0$, there is
	\begin{align*}
		&\mbP\left\{\sum_{j=1}^{m}\mbI\left(\bar{X}_j<x\right)\leq\frac{m}{2}+k\right\}\\
		=&\mbP\left\{\frac{1}{m}\sum_{j=1}^m\big\{\mbI\left(\bar{X}_j<x\right)-\mbP\left(\bar{X}_j<x\right)\big\}\leq\frac{1}{2}+\frac{k}{m}-\frac{1}{m}\sum_{j=1}^m\mbP\left(\bar{X}_j<x\right)\right\}\\
		\leq&\mbP\left\{\frac{1}{m}\sum_{j=1}^m\big\{\mbI\left(\bar{X}_j<x\right)-\mbP\left(\bar{X}_j<x\right)\big\}\leq\frac{1}{2}+\frac{k}{m}+\frac{1}{m}\sum_{j=1}^m\Big\{\frac{C_{\mathrm{B}}}{\sqrt{n_j}}-\Phi\left(\frac{\sqrt{n_j}x}{\sigma}\right)\Big\}\right\}\\
		\leq&\mbP\left\{\frac{1}{m}\sum_{j=1}^m\big\{\mbI\left(\bar{X}_j<x\right)-\mbP\left(\bar{X}_j<x\right)\big\}\leq-\sqrt{\frac{C\gamma\log p}{m}}\right\},
	\end{align*}
	where the last line uses Berry-Esseen Theorem (Lemma \ref{lemma:berry_esseen}), and $x$ is given by 
	\begin{equation*}
		x=\max_{1\leq j\leq m}\frac{\sigma}{\sqrt{n_j}}\Phi^{-1}\left(\frac{1}{2}+\frac{k}{m}+\frac{C_{\mathrm{B}}}{\sqrt{n_j}}+\sqrt{\frac{C\gamma\log p}{m}}\right)\asymp \frac{\sigma}{\sqrt{n}}\Phi^{-1}\left(\frac{1}{2}+\frac{k}{m}+\frac{C_{\mathrm{B}}}{\sqrt{n}}+\sqrt{\frac{C\gamma\log p}{m}}\right).
	\end{equation*}
	Applying Lemma \ref{lemma:cai_liu} to the i.i.d. sequence $\mbI\left(\bar{X}_j<x\right)-\mbP\left(\bar{X}_1<x\right)$, we have
	\begin{equation*}
		\mbP\left\{\frac{1}{m}\sum_{j=1}^m\mbI\left(\bar{X}_j<x\right)-\mbP\left(\bar{X}_1<x\right)\leq-\sqrt{\frac{C\gamma\log p}{m}}\right\}=O( p^{-\gamma}),
	\end{equation*}
	for some $C$ large enough. Moreover, we have the following elementary facts
	\begin{equation*}
		\big|\Phi^{-1}(x_0)\big|=\big|\Phi^{-1}(x_0)-\Phi^{-1}\left(1/2\right)\big|\leq\left|x_0-1/2\right|\left(\Phi^{-1}\right)'(x_0)\leq\frac{|x_0-1/2|}{\psi\left\{\Phi^{-1}(c_{\phi})\right\}},
	\end{equation*}
	holds for any $c_{\phi}\leq x_0<1-c_{\phi}$, where $c_{\phi}\in(0,1/2)$ is some constant. On the other hand, we know that $c_{\phi}\leq\min_j\Phi(\sqrt{n}_jx/\sigma)\leq \max_j\Phi(\sqrt{n}_jx/\sigma)<1-c_{\phi}$ holds for $m,n$ sufficiently large. Denote $C_{\delta}=1/\psi\{\Phi^{-1}(c_{\phi})\}$, then there is
	\begin{align*}
		&\mbP\left\{\sum_{j=1}^{m}\mbI\left(\bar{X}_j<\frac{C_{\delta}\sigma}{\sqrt{n}}\left(\frac{C_{\mathrm{B}}}{\sqrt{n}}+\frac{k}{m}+\sqrt{\frac{C\gamma\log p}{m}}\right)\right)\leq\frac{m}{2}+k\right\} \\
			\leq&\mbP\left\{\frac{1}{m}\sum_{j=1}^m\mbI\left(\bar{X}_j<x\right)-\mbP\left(\bar{X}_1<x\right)\leq-\sqrt{\frac{C\gamma\log p}{m}}\right\}\leq p^{-\gamma}.
	\end{align*}
	Therefore, if we choose $\tilde{C}\geq\max\{C_{\delta}C_B\sigma, C_{\delta}\sqrt{C}\sigma\}$, the bound of the first term in the left hand side of \eqref{eq:mom_lemma} is proved. By repeating the same procedure, we can prove the bound of the second term, which yields the desired result.
\end{proof}	

Before proving Theorem \ref{prop:dp_mean_sign}, we shall first prove the following proposition.
\begin{proposition}	\label{thm:priv_mean_supp}
	(Sign consistency of $\majvote$) Let $N$ i.i.d. random vectors $\{\vect{X}_1,\dots,\vect{X}_N\}$ be sampled from $\mcP(\vect{\theta}^*, C)$. Moreover, there are sufficiently large constants $C_3,C_4, \gamma_0>0$, such that
	\begin{itemize}
		\item[(a')] The dimension $p$ satisfies $p=O(n^{\gamma_0})$, and take $\lambda_N=C_3(\sqrt{\log p/N}+1/n)$;
		\item[(b')] Denote $S=\supp(\vect{\theta}^*)$, then there is $\min_{l\in S}|\theta_l^*|\geq C_4\lambda_N$.
	\end{itemize}
	Then $\bar{\vect{Q}}(\mbX)$ defined in \eqref{eq:med-mean_sign} satisfies that, for some large $\gamma_2$ depends on $C_1,C_2,\gamma_0$, there is
	\begin{equation*}
		\mbP\Big(\bar{\vect{Q}}(\mbX)=\sgn(\vect{\theta}^*)\Big)\geq 1-O(p^{-\gamma_2}).
	\end{equation*}	
\end{proposition}
\begin{proof}[Proof of Proposition \ref{thm:priv_mean_supp}]
	If we can show that
	\begin{equation}	\label{eq1:thm_mean_supp}
		\max_{1\leq l\leq p}\mbP\Big(\bar{Q}_l(\mbX)\neq\sgn(\theta_l^*)\Big)=O(p^{-\gamma}),
	\end{equation} 
	where $\gamma>0$ is large enough. Then this theorem is proved as follows
	\begin{align*} 
		&1 - \mbP\Big(\bar{\vect{Q}}(\mbX)=\sgn(\vect{\theta}^*)\Big)\\
		\leq& p\max_{1\leq l\leq p}\mbP\Big(\bar{Q}_l(\mbX)\neq\sgn(\theta_l^*)\Big)= O(pp^{-\gamma})= O(p^{-\gamma+\gamma_0}),
	\end{align*}
	provided that $\gamma>\gamma_0$ and sufficient large $\gamma_0$. For the $l$-th coordinate, we firstly suppose that $\theta_l^*=0$. Since $\vect{X}_i$'s are sampled from the distribution \eqref{eq:sample_dist}, by Lemma \ref{lemma:concen_mom} with $k=0$, if 
	\begin{equation*}
		\lambda_N= C_0\left(\sqrt{\frac{\log p}{mn}}+\frac{1}{n}\right)\geq c_{\gamma,0},
	\end{equation*}
	we have
	\begin{align*}
		&\mbP\Big(\bar{Q}_l(\mbX)\neq\sgn(\theta_l^*)\Big)\\
		\leq &\mbP\Big\{\sum_{j=1}^{m}\mbI\Big(\bar{X}_{j,l}<\lambda_N\Big)\leq\frac{m}{2}\Big\}+\mbP\Big\{\sum_{j=1}^{m}\mbI\Big(\bar{X}_{j,l}>-\lambda_N\Big)\leq\frac{m}{2}\Big\}=O(p^{-\gamma}).
	\end{align*}
	Next we assume $\theta_l^*>0$. We use Lemma \ref{lemma:concen_mom} again, if 
	\begin{equation*}
		\theta_l^*\geq C_1\Big(\sqrt{\frac{\log p}{mn}}+\frac{1}{n}\Big)\geq c_{\gamma,0}+\lambda_N,
	\end{equation*}
	then there is
	\begin{align*}
		&\mbP\Big(\bar{Q}_l(\mbX)\neq\sgn(\theta_l^*)\Big)\leq \mbP\Big\{\sum_{j=1}^{m}\mbI\Big(\bar{X}_{j,l}>\lambda_N\Big)\leq\frac{m}{2}\Big\}\\
		\leq &\mbP\Big\{\sum_{j=1}^{m}\mbI\Big(\bar{X}_{j,l}-\theta_l^*>-c_{\gamma,0}\Big)\leq\frac{m}{2}\Big\}=O(p^{-\gamma}).
	\end{align*}
	Lastly, when $\theta_l^*<0$, the proof is the same as above. Therefore \eqref{eq1:thm_mean_supp} is proved.
\end{proof}

\begin{proof}[Proof of Theorem \ref{prop:dp_mean_sign}]
	By Theorem \ref{thm:priv_mean_supp}, we know that $\bar{\vect{Q}}(\mbX)=\sgn(\vect{\theta}^*)$ holds with probability not less than $1-p^{-\gamma_1}$. To prove the sign consistency of $\hat{\vect{Q}}(\mbX)$, from Theorem \ref{thm:medsign_cons} we know that it is enough to show for each $l\in S$, there is
	\begin{equation*}
		f^S(\vect{Q}^r_l) \geq 8(\gamma+1)\sqrt{2\tilde{s}\log(2/\delta)}\log p/\epsilon
	\end{equation*}
	hold with probability $1-O(p^{-\gamma})$. For $l\in S$, without loss of generality assume that $\sgn(\theta^*_l)=\sgn(\bar{Q}_l(\mbX))=1$, then
	\begin{align*}
		f_l(\vect{Q}^r_l,\bar{Q}_l(\mbX)) =& N_l^+-N_l^--N_l^0=2N_l^+ - m\\
		=&2\sum_{j=1}^m\mbI\Big(\bar{X}_{j,l}>\lambda_N\Big)-m\\
		=&2\Big\{\sum_{j=1}^m\mbI\Big(\bar{X}_{j,l}-\theta_l^*>\lambda_N-\theta_l^*\Big)-\frac{m}{2}\Big\}.
	\end{align*} 
	Therefore,
	\begin{align*}
		&\mbP\Big(f^S(\vect{Q}_l^r)\geq 8(\gamma+1)\sqrt{2\tilde{s}\log(2/\delta)}\log p/\epsilon\Big)\\
		=& \mbP\left(\sum_{j=1}^m\mbI\Big(\bar{X}_{j,l}-\theta_l^*>\lambda_N-\theta_l^*\Big)-\frac{m}{2}\geq 4(\gamma+1)\sqrt{2\tilde{s}\log(2/\delta)}\log p/\epsilon\right)\\
		\geq& 1-O(p^{-\gamma_1}),
	\end{align*}
	as long as 
	\begin{align*}
		\theta_l^*-\lambda_N\geq\frac{C_{\delta}\sigma}{\sqrt{n}}\left(\frac{C_{\mathrm{B}}}{\sqrt{n}}+\frac{4(\gamma+1)\sqrt{2\tilde{s}\log(2/\delta)}\log p}{m\epsilon}+\sqrt{\frac{C\gamma_1\log p}{m}}\right),
	\end{align*}
	which proves the first part of the proposition. Similarly, when $\tilde{s}>s$,for $l\notin S$, there is
	\begin{align*}
		f_l(\vect{Q}^r_l,0) =& \min\{N_l^0-N_l^++N_l^-,N_l^0+N_l^+-N_l^-\}\\
		=&\min\Big\{2\sum_{j=1}^m\mbI\Big(\bar{X}_{j,l}\leq\lambda_N\Big),2\sum_{j=1}^m\mbI\Big(\bar{X}_{j,l}\geq-\lambda_N\Big)\Big\}-m.
	\end{align*} 
	Therefore 
	\begin{align*}
		&\mbP\Big(2\sum_{j=1}^m\mbI\Big(\bar{X}_{j,l}\leq\lambda_N\Big)-m\geq 8(\gamma+1)\sqrt{2\tilde{s}\log(2/\delta)}\log p/\epsilon\Big)\\
		=& \mbP\left(\sum_{j=1}^m\mbI\Big(\bar{X}_{j,l}\leq\lambda_N\Big)-\frac{m}{2}\geq 4(\gamma+1)\sqrt{2\tilde{s}\log(2/\delta)}\log p/\epsilon\right)\\
		\geq& 1-O(p^{-\gamma_1}),
	\end{align*}
	as long as 
	\begin{align*}
		\lambda_N\geq\frac{C_{\delta}\sigma}{\sqrt{n}}\left(\frac{C_{\mathrm{B}}}{\sqrt{n}}+\frac{4(\gamma+1)\sqrt{2\tilde{s}\log(2/\delta)}\log p}{m\epsilon}+\sqrt{\frac{C\gamma_1\log p}{m}}\right).
	\end{align*}
	Similarly, we can show that
	\begin{align*}
		&\mbP\Big(2\sum_{j=1}^m\mbI\Big(\bar{X}_{j,l}\geq-\lambda_N\Big)-m\geq 8(\gamma+1)\sqrt{2\tilde{s}\log(2/\delta)}\log p/\epsilon\Big)\geq 1-O(p^{-\gamma_1}).
	\end{align*}
	Therefore we have
	\begin{align*}
	    &\mbP\Big(f_l(\vect{Q}^r_l,0) \geq 8(\gamma+1)\sqrt{2\tilde{s}\log(2/\delta)}\log p/\epsilon\Big)\geq 1-O(p^{-\gamma_1}),
	\end{align*}
	and Theorem \ref{thm:medsign_cons} applies.
\end{proof}


\subsection{Proof of Results in Section \ref{sec:priv_lasso_theory}}\label{sec:proof_4-2}

\begin{lemma}	\label{eq:cov_bound}
	Let $\vect{X}_1,\dots,\vect{X}_n$ be i.i.d. random vectors sampled from the distribution in \eqref{eq:cov_dist}. Denote its covariance matrix as $\vect{\Sigma}$, and the sample covariance matrix as $\hat{\vect{\Sigma}}$. Then for every $\gamma>1$, there exists a constant $\tilde{C}>0$ such that
	\begin{equation*}
		\mbP\Big(\Norm{\vect{\Sigma}^{-1}\hat{\vect{\Sigma}}-\vect{I}}_{\infty}\geq\tilde{C}\sqrt{\frac{\log p}{n}}\Big)=O(p^{-\gamma}).
	\end{equation*}
\end{lemma}

\begin{proof}
	Recall that the inverse covariance matrix is denoted as $\vect{\Sigma}^{-1}=(\vect{\omega}_1,\dots,\vect{\omega}_p)$, and $\vect{e}_l$ as the $l$-th coordinate vector. Then the $(l_1,l_2)$-entry of the matrix $\vect{\Sigma}^{-1}\hat{\vect{\Sigma}}-\vect{I}$ is
	\begin{equation*}
	\begin{aligned}
		(\vect{\Sigma}^{-1}\hat{\vect{\Sigma}}-\vect{I})_{l_1,l_2}=\frac{1}{n}\sum_{i=1}^n\vect{\omega}_{l_1}^{\tp}\vect{X}_i\cdot \vect{e}_{l_2}^{\tp}\vect{X}_i-\delta_{l_1,l_2}.
	\end{aligned}
	\end{equation*}
	Since the dimension $p$ is assumed to be bounded, and the covariance matrix is positive definite, there exist a constant $\rho\in(0,1)$ such that
	\begin{equation*}
		\rho\leq\Lambda_{\min}(\vect{\Sigma})\leq\Lambda_{\max}(\vect{\Sigma})\leq\rho^{-1}.
	\end{equation*}
	Then we have
	\begin{equation*}
		\max_{1\leq l\leq p}|\vect{\omega}_l|_2\leq\Norm{\vect{\Sigma}^{-1}}\leq \rho^{-1}.
	\end{equation*}
	Since $\vect{X}$ is sub-Gaussian in \eqref{eq:cov_dist}, we obtain
	\begin{align*}
		&\max_{1\leq l_1,l_2\leq p}\mbE\big\{\exp\eta_1\rho\big|\vect{\omega}_{l_1}^{\tp}\vect{X}_i\cdot \vect{e}_{l_2}^{\tp}\vect{X}_i-\delta_{l_1,l_2}\big|\big\}\\
		\leq&e^{\eta_1\rho}\cdot\sup_{|\vect{v}|_2\leq 1}\mbE\big\{\eta_1|\vect{v}^{\tp}\vect{X}|^2\big\}\leq e^{\eta_1\rho}C_1.
	\end{align*}
	Therefore, we can apply Lemma \ref{lemma:cai_liu} to each coordinate and yield
	\begin{align*}
		&\mbP\Big(\Norm{\vect{\Sigma}^{-1}\hat{\vect{\Sigma}}-\vect{I}}_{\infty}\geq \tilde{C}\sqrt{\frac{\log p}{n}}\Big)\\
		\leq& p^2\max_{1\leq l_1,l_2\leq p}\mbP\Big(\Big|\frac{1}{n}\sum_{i=1}^n\vect{\omega}_{l_1}^{\tp}\vect{X}_i\cdot \vect{e}_{l_2}^{\tp}\vect{X}_i-\delta_{l_1,l_2}\Big|\geq \tilde{C}\sqrt{\frac{\log p}{n}}\Big)\\
		= & O(p^2p^{-\gamma-2})=O(p^{-\gamma}),
	\end{align*}
	for some $\tilde{C}$ sufficiently large. Therefore, the lemma is proved.
\end{proof}
	
Before proving Theorem \ref{prop:dp_lasso_sign}, we first provide the sign consistency of the non-private $\majvote$ Lasso estimator. Namely, we define
\begin{equation}    \label{eq:med-lasso_sign}
    \bar{\vect{Q}}(\mbX) = \majvote(\mcQ_0(\hat{\vect{\theta}}_j)|1\leq j\leq m).
\end{equation}
Then we have the following proposition.

\begin{proposition}	\label{thm:priv_reg_supp}
	(Sign consistency of $\majvote$ Lasso) Let $N$ i.i.d. random vectors $\mbX=\{(\vect{X}_1,Y_1)$, $\dots,(\vect{X}_N,Y_N)\}$ be sampled from $\mcP_{\vect{X},Y}(\vect{\theta}^*,\rho,\eta_1,C_1,\eta_2,C_2)$. Moreover, we assume the same conditions as in Theorem \ref{prop:dp_lasso_sign} holds. Then $\bar{\vect{Q}}(\mbX)$ defined in \eqref{eq:med-lasso_sign} satisfies that, there is constant $\gamma_2>0$ such that
	\begin{equation*}
		\mbP\Big(\bar{\vect{Q}}(\mbX)=\sgn(\vect{\theta}^*)\Big)\geq 1-O(p^{-\gamma_2}).
	\end{equation*}	
\end{proposition}

\begin{proof}[Proof of Proposition \ref{thm:priv_reg_supp}]
	For each $j\in\{1,\dots,m\}$, taking sub-gradient of \eqref{eq:local_Lasso} at $\hat{\vect{\theta}}_j$, we have that
	\begin{equation*}
		-\frac{1}{n_j}\sum_{i\in\mcH_j}(Y_i-\vect{X}_i^{\tp}\hat{\vect{\theta}}_j)\vect{X}_i+\lambda_j\vect{Z}_j=0,
	\end{equation*}
	where $\vect{Z}_j$ is the sub-gradient satisfying $|\vect{Z}_j|_{\infty}\leq1$. Rearranging the terms and multiplying $\vect{\Sigma}^{-1}$ on the both sides, we have
	\begin{equation}	\label{eq1:thm_reg_supp}
		\hat{\vect{\theta}}_j-\vect{\theta}^*=(\vect{I}-\vect{\Sigma}^{-1}\frac{1}{n_j}\sum_{i\in\mcH_j}\vect{X}_i\vect{X}_i^{\tp})(\hat{\vect{\theta}}_j-\vect{\theta}^*)+\frac{1}{n_j}\sum_{i\in\mcH_j}\vect{\Sigma}^{-1}\vect{X}_iz_i-\lambda_j\vect{\Sigma}^{-1}\vect{Z}_j.
	\end{equation}
	Taking $\tilde{\eta}=\min\{\eta_1\rho,\eta_2\}$, since the noise and covariates are assumed to be sub-Gaussian in \eqref{eq:cov_dist}, for each coordinate $l\in \{1,\dots, p\}$, we have
	\begin{align*}
		&\max_{1\leq l\leq p}\mbE\big\{\exp\tilde{\eta}|\vect{\omega}_l\vect{X}\cdot z|\big\}\\
		\leq&\max_{1\leq l\leq p}\mbE\Big\{\exp\Big(\frac{1}{2}\eta\rho|\vect{\omega}_l\vect{X}|^2+\frac{1}{2}\eta_2|z|^2\Big)\Big\}\\
		\leq&\max_{1\leq l\leq p}\Big[\mbE\big\{\exp\eta\rho|\vect{\omega}_l\vect{X}|^2\big\}\cdot\mbE\big\{\exp\eta_2|z|^2\big\}\Big]^{1/2}\leq \sqrt{C_1C_2}.
	\end{align*}
	Therefore by Lemma \ref{lemma:cai_liu}, we know that there exists a constant $\tilde{C}_1>0$ such that
	\begin{equation}	\label{eq:thm3_keybound1}
		\max_{1\leq j\leq m}\Big|\frac{1}{n_j}\sum_{i\in\mcH_j}\vect{\Sigma}^{-1}\vect{X}_iz_i\Big|_{\infty}\leq \tilde{C}_1\sqrt{\frac{\log p}{n}},
	\end{equation}
	with probability larger than $1-O(p^{-\gamma})$. On the other hand, by the fact that $|\vect{Z}_j|_{\infty}\leq 1$, we have 
	\begin{equation}	\label{eq:thm3_keybound2}
		|\lambda_j\vect{\Sigma}^{-1}\vect{Z}_j|_{\infty}\leq \lambda_j\Norm{\vect{\Sigma}^{-1}}_{L_{\infty}}.
	\end{equation}
	Moreover, since $|\hat{\vect{\theta}}_j|_0\leq\tilde{s}$ by Lemma \ref{eq:cov_bound}, we know
	\begin{equation}	\label{eq:thm3_keybound3}
		\max_{1\leq j\leq m}\Big|\Big(\vect{I}-\vect{\Sigma}^{-1}\frac{1}{n_j}\sum_{i\in\mcH_j}\vect{X}_i\vect{X}_i^{\tp}\Big)(\hat{\vect{\theta}}_j-\vect{\theta}^*)\Big|_{\infty}=O_{\mbP}\Big(\tilde{s}\sqrt{\frac{\log p}{n}}|\hat{\vect{\theta}}_j-\vect{\theta}^*|_{\infty}\Big)\leq \frac{1}{2}|\hat{\vect{\theta}}_j-\vect{\theta}^*|_{\infty}.
	\end{equation}
	By substituting equations \eqref{eq:thm3_keybound1}, \eqref{eq:thm3_keybound2}, and \eqref{eq:thm3_keybound3} into \eqref{eq1:thm_reg_supp}, we have
	\begin{equation*}
		|\hat{\vect{\theta}}_j-\vect{\theta}^*|_{\infty}\leq 2\lambda_j\Norm{\vect{\Sigma}^{-1}}_{L_{\infty}}+2\tilde{C}_1\sqrt{\frac{\log p}{n}}.
	\end{equation*}
	From \eqref{eq1:thm_reg_supp}, the $l$-th coordinate can be written in the following form
	\begin{equation}	\label{eq2:thm_reg_supp}
		\hat{\theta}_{j,l}-\theta^*_l=-\lambda_j\omega_{l,l}Z_{j,l}-\lambda_j\vect{\omega}^{\tp}_{l,-l}\vect{Z}_{j,-l}+\frac{1}{n_j}\sum_{i\in\mcH_j}\vect{\omega}_l^{\tp}\vect{X}_iz_i+o_{\mbP}\Big(\frac{\log p}{n}\Big).
	\end{equation}
	It left to rehash the argument in the proof of Theorem \ref{thm:priv_mean_supp}. From Lemma \ref{lem:theta_sign} below we have that
	\begin{align*}
		&1-\mbP\Big(\bar{\vect{Q}}(\mbX)=\sgn(\vect{\theta}^*)\Big)\leq p\cdot\max_{1\leq l\leq p}\mbP\Big(\bar{Q}_l(\mbX)\neq\sgn(\theta_l^*)\Big)=O(p^{-\gamma+1}).
	\end{align*}
	Therefore we proved Proposition \ref{thm:priv_reg_supp}.
\end{proof}

\begin{lemma}	\label{lem:theta_sign}
	Assume the same conditions in Proposition \ref{thm:priv_reg_supp}. For every $1\leq l\leq p$, we have
	\begin{equation*}
		\mbP\Big(\bar{Q}_l(\mbX)\neq\sgn(\theta_l^*)\Big)=O(p^{-\gamma}),
	\end{equation*}
	for arbitrarily fixed $\gamma>1$.
\end{lemma}

\begin{proof}
	When $\theta^*_{l}=0$, we know that
	\begin{equation*}	
	\mbP\Big(\bar{Q}_l(\mbX)\neq\sgn(\theta_l^*)\Big)\leq\mbP\Big\{\sum_{j=1}^{m}\mbI\Big(\hat{\theta}_{j,l}< 0\Big)\geq\frac{m}{2}\Big\}+\mbP\Big\{\sum_{j=1}^{m}\mbI\Big(\hat{\theta}_{j,l}>0\Big)\geq\frac{m}{2}\Big\}.
	\end{equation*} 
	By symmetry of the formulation, it is enough to prove
	\begin{equation}	\label{eq:lem:theta_sign}
		\mbP\Big\{\sum_{j=1}^m\mbI\Big(\hat{\theta}_{j,l}>0\Big)\geq\frac{m}{2}\Big\}=O(p^{-\gamma}).
	\end{equation}
	When $\hat{\theta}_{j,l}>0$, we know that $Z_{j,l}=1$ (see \eqref{eq2:thm_reg_supp}). Moreover, by condition (c), we have
	\begin{align*}
	    0<\hat{\theta}_{j,l}\leq&-\lambda_j\omega_{l,l}+\lambda_j|\vect{\omega}_{l,-l}|_{1}+\frac{1}{n_j}\sum_{i\in\mcH_j}\vect{\omega}_l^{\tp}\vect{X}_iz_i+o_{\mbP}\Big(\frac{\log p}{n}\Big)\\
	    \leq&-\Delta_0\lambda_j\omega_{l,l}+\frac{1}{n_j}\sum_{i\in\mcH_j}\vect{\omega}_l^{\tp}\vect{X}_iz_i+o_{\mbP}\Big(\frac{\log p}{n}\Big)\\
	    \leq&-\frac{1}{2}\Delta_0\lambda_N\omega_{l,l}+\frac{1}{n_j}\sum_{i\in\mcH_j}\vect{\omega}_l^{\tp}\vect{X}_iz_i.
	\end{align*}
	Therefore from \eqref{eq2:thm_reg_supp} we have that
	\begin{equation}	\label{eq:thetal_zero}
	\begin{aligned}
		\mbP\Big\{\sum_{j=1}^m\mbI\Big(\hat{\theta}_{j,l}> 0\Big)\geq\frac{m}{2}\Big\}\leq\mbP\Big\{\sum_{j=1}^m\mbI\Big(\frac{1}{n_j}\sum_{i\in\mcH_j}\vect{\omega}_l^{\tp}\vect{X}_iz_i\geq\frac{1}{2}\Delta_0\lambda_N\omega_{l,l}\Big)\geq\frac{m}{2}\Big\}.
	\end{aligned}
	\end{equation}
	Applying Lemma \ref{lemma:concen_mom} with $k=0$ to the i.i.d. random variables $\vect{\omega}_l^{\tp}\vect{X}_iz_i$ we can prove \eqref{eq:lem:theta_sign} by taking
	\begin{equation}	\label{eq:lem:theta_sign_lambda}
		\lambda_N\geq\frac{2\tilde{C}_1}{\Delta_0\omega_{l,l}}\Big(\sqrt{\frac{\log p}{mn}}+\frac{1}{n}\Big),
	\end{equation}
	with $\tilde{C}$ sufficiently large. Repeat the argument for the other half, we can prove the case for $\theta^*_{l}=0$.
	
	When $\theta^*_{l}>0$, again from equation \eqref{eq2:thm_reg_supp}, we have
	\begin{equation*}
		\hat{\theta}_{j,l}\geq\theta^*_l+\frac{1}{n_j}\sum_{i\in\mcH_j}\vect{\omega}_l^{\tp}\vect{X}_iz_i-\Norm{\vect{\Sigma}^{-1}}_{L_{\infty}}\lambda_j+o_{\mbP}\Big(\frac{\log p}{n}\Big).
	\end{equation*}
	Therefore
	\begin{equation}	\label{eq:thetal_nonzero}
	\begin{aligned}
	&\mbP\Big(\bar{Q}_l(\mbX)\neq\sgn(\theta_l^*)\Big)\leq\mbP\Big\{\sum_{j=1}^{m}\mbI\Big(\hat{\theta}_{j,l}\leq0\Big)\geq\frac{m}{2}\Big\}\\
	\leq &\mbP\Big\{\sum_{j=1}^{m}\mbI\Big(-\frac{1}{n_j}\sum_{i\in\mcH_j}\vect{\omega}_l^{\tp}\vect{X}_iz_i\geq \theta^*_l-2\Norm{\vect{\Sigma}^{-1}}_{L_{\infty}}\lambda_j \Big)\geq\frac{m}{2}\Big\}.
	\end{aligned}
	\end{equation}
	Applying Lemma \ref{lemma:concen_mom} we can show that
	\begin{equation*}
		\mbP\Big\{\sum_{j=1}^{m}\mbI\Big(-\frac{1}{n_j}\sum_{i\in\mcH_j}\vect{\omega}_l^{\tp}\vect{X}_iz_i\geq \theta^*_l-2\Norm{\vect{\Sigma}^{-1}}_{L_{\infty}}\lambda_j \Big)\geq\frac{m}{2}\Big\}=O(p^{-\gamma}),
	\end{equation*}
	provided that
	\begin{equation*}
		\theta^*_l-2\Norm{\vect{\Sigma}^{-1}}_{L_{\infty}}\lambda_j\geq \tilde{C}_2\Big(\sqrt{\frac{\log p}{mn}}+\frac{1}{n}\Big),
	\end{equation*}
	for all $j\in\{1,\dots,m\}$ with $\tilde{C}_2$ sufficiently large. Therefore we have that
	\begin{equation*}
		\theta_l^*\geq \tilde{C}_3\Big(\sqrt{\frac{\log p}{mn}}+\frac{1}{n}+\max_{1\leq j\leq m}\lambda_j\Big),
	\end{equation*}
	for some $\tilde{C}_3>0$. When $\theta^*_{l}<0$, the prove is essentially the same as above, hence we omit it for brevity. Thus the lemma is proved.
\end{proof}

\begin{proof}[Proof of Theorem \ref{prop:dp_lasso_sign}]
	Similarly as the proof of Proposition \ref{prop:dp_mean_sign}, for $l\in S$ such that $\sgn(\theta_l^*)=\sgn(\bar{Q}_l(\mbX))=1$, we wish there holds
	\begin{equation*}
		f^S(\vect{Q}^r_l) = 2\sum_{j=1}^m\mbI\Big(\hat{\theta}_{j,l}>0\Big)-m\geq 8(\gamma+1)\sqrt{2\tilde{s}\log(2/\delta)}\log p/\epsilon.
	\end{equation*}
	By \eqref{eq:thetal_nonzero} we know
	\begin{align*}
		&\mbP\Big(f^S(\vect{Q}_l^r)\geq 8(\gamma+1)\sqrt{2\tilde{s}\log(2/\delta)}\log p/\epsilon\Big)\\
		=& \mbP\left(\sum_{j=1}^m\mbI\Big(\frac{1}{n}\sum_{i\in\mcH_j}\vect{\omega}_l^{\tp}\vect{X}_iz_i>\theta^*_l-2\Norm{\vect{\Sigma}^{-1}}_{L_{\infty}}\lambda_j\Big)-\frac{m}{2}\geq 4(\gamma+1)\sqrt{2\tilde{s}\log(2/\delta)}\log p/\epsilon\right)\\
		\geq& 1-O(p^{-\gamma_1}),
	\end{align*}
	as long as
	\begin{align*}
		\theta_l^*-2\Norm{\vect{\Sigma}^{-1}}_{L_{\infty}}\max_{1\leq j\leq m}\lambda_j\geq \tilde{C}\left(\frac{1}{n}+\frac{4(\gamma+1)\sqrt{2\tilde{s}\log(2/\delta)}\log p}{m\sqrt{n}\epsilon}+\sqrt{\frac{\log p}{mn}}\right),
	\end{align*}
	which proves the first part of the proposition. Similarly, when $\tilde{s}>s$,for $l\notin S$, there is
	\begin{align*}
		f_l(\vect{Q}^r_l,0)=\min\Big\{2\sum_{j=1}^m\mbI\Big(\hat{\theta}_{j,l}\leq 0\Big),2\sum_{j=1}^m\mbI\Big(\hat{\theta}_{j,l}\geq 0\Big)\Big\}-m.
	\end{align*} 
	Therefore by \eqref{eq:thetal_zero},
	\begin{align*}
		&\mbP\Big(f_l(\vect{Q}_l^r,0)\geq 8(\gamma+1)\sqrt{2\tilde{s}\log(2/\delta)}\log p/\epsilon\Big)\\
		\geq&1- \mbP\left(\sum_{j=1}^m\mbI\Big(\frac{1}{n_j}\sum_{i\in\mcH_j}\vect{\omega}_l^{\tp}\vect{X}_iz_i\leq\frac{1}{2}\Delta_0\lambda_N\omega_{l,l}\Big)-\frac{m}{2}\leq 4(\gamma+1)\sqrt{2\tilde{s}\log(2/\delta)}\log p/\epsilon\right)\\
		&-\mbP\left(\sum_{j=1}^m\mbI\Big(\frac{1}{n_j}\sum_{i\in\mcH_j}\vect{\omega}_l^{\tp}\vect{X}_iz_i\geq-\frac{1}{2}\Delta_0\lambda_N\omega_{l,l}\Big)-\frac{m}{2}\leq 4(\gamma+1)\sqrt{2\tilde{s}\log(2/\delta)}\log p/\epsilon\right)\\
		\geq& 1-O(p^{-\gamma_1}),
	\end{align*}
	as long as 
	\begin{align*}
		\lambda_N\geq \tilde{C}\left(\frac{1}{n}+\frac{4(\gamma+1)\sqrt{2\tilde{s}\log(2/\delta)}\log p}{m\sqrt{n}\epsilon}+\sqrt{\frac{\log p}{mn}}\right).
	\end{align*}
\end{proof}

\end{document}

%% file: notations.tex
\newcommand{\Norm}[1]{\left\Vert #1\right\Vert}			
\newcommand{\vect}[1]{\boldsymbol{#1}}				

\newcommand{\tp}{{\rm T}}						
\newcommand{\sgn}{{\rm sgn}}						
\newcommand{\med}{{\rm med}}					
\newcommand{\argmin}[1]{\mathop{\rm{argmin}}_{#1}}	
\newcommand{\argmax}[1]{\mathop{\rm{argmax}}_{#1}}	
											
\newcommand*\widebar[1]{\,						
	\hbox{%
   		\kern0.01em%
		\vbox{%
			\hrule height 0.5pt		
			\kern0.33ex			
			\hbox{%
				\kern-0.1em		
				\ensuremath{#1}%
				\kern-0.1em		
			}%
		}%
	}%
\,}%
\let\hat\widehat
\let\tilde\widetilde
\let\bar\widebar

\newcommand{\mcA}{{\mathcal A}}

\newcommand{\mcH}{{\mathcal H}}					

\newcommand{\mcN}{{\mathcal N}}					

\newcommand{\mcP}{{\mathcal P}}
\newcommand{\mcQ}{{\mathcal Q}}

\newcommand{\mcX}{{\mathcal X}}

\newcommand{\mbB}{{\mathbb B}}					

\newcommand{\mbE}{{\mathbb E}}					

\newcommand{\mbI}{{\mathbb I}}					

\newcommand{\mbP}{{\mathbb P}}					
\newcommand{\mbQ}{{\mathbb Q}}					
\newcommand{\mbR}{{\mathbb R}}					
\newcommand{\mbS}{{\mathbb S}}					

\newcommand{\mbX}{{\mathbb X}}

\DeclareMathAlphabet\EuScriptBF{U}{eus}{b}{n}


%% file: plot/med-dc_intuition.tex
\begin{figure}[t]
\begin{center}
\begin{tikzpicture}[
	dot/.style={draw,circle,minimum size=3mm,inner sep=0pt,outer sep=0pt,fill=black},
	dots/.style={draw,circle,minimum size=3mm,inner sep=0pt,outer sep=0pt,fill=white},
	to/.style={->,>=stealth',shorten >=1pt,thick,font=\sffamily\large},
	]
	\node[dot] at (-3.75,7.75) {};
	\node[dot] at (-3.75,7.25) {};
	\node[dot] at (-3.75,6.75) {};
	\node[dot] at (-3.75,6.25) {};
	\node[dots] at (-3.75,5.75) {};	
	\node[dots] at (-3.75,5.25) {};
	\node[dot] at (-3.75,4.75) {};
	\node[dot] at (-3.75,3.25) {};
	\node[dot] at (-3.75,1.25) {};
	\node[dots] at (-3.75,0.75) {};
	
	\node[dot] at (-2.75,7.75) {};
	\node[dot] at (-2.75,7.25) {};
	\node[dot] at (-2.75,6.25) {};
	\node[dots] at (-2.75,5.75) {};
	\node[dots] at (-2.75,5.25) {};
	\node[dots] at (-2.75,4.75) {};
	\node[dots] at (-2.75,4.25) {};
	\node[dot] at (-2.75,2.75) {};
	\node[dot] at (-2.75,1.25) {};
	\node[dot] at (-2.75,0.25) {};
	
	\node[dot] at (-1.75,7.75) {};
	\node[dot] at (-1.75,7.25) {};
	\node[dot] at (-1.75,6.75) {};
	\node[dot] at (-1.75,6.25) {};
	\node[dots] at (-1.75,5.75) {};
	\node[dots] at (-1.75,4.75) {};
	\node[dot] at (-1.75,3.75) {};
	\node[dots] at (-1.75,3.25) {};
	\node[dots] at (-1.75,0.25) {};

	\node[dot] at ( 0.25,7.25) {};
	\node[dot] at ( 0.25,6.75) {};
	\node[dot] at ( 0.25,6.25) {};
	\node[dots] at ( 0.25,5.75) {};
	\node[dots] at ( 0.25,5.25) {};
	\node[dots] at ( 0.25,4.75) {};
	\node[dots] at ( 0.25,1.25) {};
	\node[dot] at ( 0.25,2.25) {};
	\node[dot] at ( 0.25,0.75) {};
	
	\node[dot] at ( 3.25,7.75) {};
	\node[dot] at ( 3.25,7.25) {};
	\node[dot] at ( 3.25,6.75) {};
	\node[dot] at ( 3.25,6.25) {};
	\node[dots] at ( 3.25,5.75) {};
	\node[dots] at ( 3.25,5.25) {};
	\node[dots] at ( 3.25,4.75) {};

	\draw [thick] [decorate,decoration={brace,amplitude=10pt},xshift=-4pt,yshift=0pt]
(-4,0) -- (-4,4.46) node [black,midway,xshift=-0.6cm] {\footnotesize $S^c$} ;
	\draw [thick] [decorate,decoration={brace,amplitude=10pt},xshift=-4pt,yshift=0pt]
(-4,4.54) -- (-4,5.96) node [black,midway,xshift=-0.6cm] {\footnotesize $S^-$} ;
	\draw [thick] [decorate,decoration={brace,amplitude=10pt},xshift=-4pt,yshift=0pt]
(-4,6.04) -- (-4,8) node [black,midway,xshift=-0.6cm] {\footnotesize $S^+$} ;
	\draw [thick, dashed] (-4, 6) -- (0.5,6);
	\draw [thick, dashed] ( 3, 6) -- (3.5,6);
	\draw [thick, dashed] (-4, 4.5) -- (0.5,4.5);
	\draw [thick, dashed] ( 3, 4.5) -- (3.5,4.5);
	\draw [thick] (-4,0) -- (-4,8);
	\draw [thick] (-3.5,0) -- (-3.5,8);
	\draw [thick] (-3,0) -- (-3,8);
	\draw [thick] (-2.5,0) -- (-2.5,8);
	\draw [thick] (-2,0) -- (-2,8);
	\draw [thick] (-1.5,0) -- (-1.5,8);
	\draw [thick] (0,0) -- (0,8);
	\draw [thick] (0.5,0) -- (0.5,8);
	\draw [thick] (3,0) -- (3,8);
	\draw [thick] (3.5,0) -- (3.5,8);
	\draw (-3.75, -0.6) node[above] {$\vect{Q}^c_1$};
	\draw (-2.75, -0.6) node[above] {$\vect{Q}^c_2$};
	\draw (-1.75, -0.6) node[above] {$\vect{Q}^c_3$};
	\draw ( 0.25, -0.6) node[above] {$\vect{Q}^c_m$};
	\draw ( 3.25, -0.6) node[above] {$\tilde{\vect{Q}}$};
	\draw (-0.75,     4) node[above] {$\bullet\;\bullet\;\bullet$};
	
	\draw[to] (0.75, 4)  to[bend right=0]
		node[above] {$\majvote$} (2.75, 4);
\end{tikzpicture}
\end{center}
	\caption{This figure visualizes the mechanism of the majority vote approach. Denote $S^+$, $S^-$ and $S^c$ as the sets of positives, negatives, and zeros of the true parameter $\vect{\theta}^*$, respectively. The black dots and white dots in each column represent the estimated positive and negative locations on each local machine. By aggregating these local sign vectors with the proposed method, we can recover the true signs with high probability.}
	\label{fig:med-dc_intuition}
\end{figure}